\documentclass[letterpaper,11pt]{article}

\usepackage[utf8]{inputenc} 
\usepackage[T1]{fontenc}    
\usepackage{hyperref}       
\usepackage{url}            
\usepackage{booktabs}       
\usepackage{amsfonts}       
\usepackage{nicefrac}       
\usepackage{microtype}      
\usepackage{xcolor}         

\usepackage{enumitem}
\usepackage{multirow}

\usepackage{amsmath}
\usepackage{amssymb}
\usepackage{amsthm}
\usepackage[margin=1in]{geometry}

\usepackage{natbib}

\usepackage{thmtools,thm-restate}
\newtheorem{theorem}{Theorem}

\newtheorem{lemma}[theorem]{Lemma}

\newtheorem{definition}{Definition}

\newtheorem{setup}{Setup}

\newtheorem{assumption}{Assumption}

\definecolor{DarkGreen}{rgb}{0.1,0.5,0.1}
\definecolor{DarkRed}{rgb}{0.5,0.1,0.1}
\definecolor{DarkBlue}{rgb}{0.1,0.1,0.5}
\definecolor{Gray}{rgb}{0.2,0.2,0.2}

\DeclareMathOperator*{\argmax}{arg\,max}

\newcommand{\Indicator}{1}
\usepackage{color-edits}[]

\newcommand{\vectorfont}[1]{\mathbf{#1}}

\newcommand{\DPrior}{\mathcal{D}^{\text{Prior}}}
\newcommand{\DNoise}{\mathcal{D}^{\text{Noise}}}
\newcommand{\DPosterior}{\mathcal{D}^{\text{Post}}}
\newcommand{\IS}{\mathcal{I}}
\newcommand{\ISShared}{\mathcal{I}_{\textrm{shared}}}

\newcommand{\AAll}{\mathcal{A}_{\text{all}}}
\newcommand{\AAllContinuous}{\mathcal{A}_{\text{all}}^{\text{cont}}}

\newcommand{\U}{\mathcal{U}}
\newcommand{\USeparate}{\mathcal{U}^{\text{separate}}}
\newcommand{\UShared}{\mathcal{U}^{\text{shared}}}
\addauthor{nh}{magenta}

\definecolor{mydarkblue}{rgb}{0,0.08,0.45}
\hypersetup{ %
  colorlinks=true,
  linkcolor=mydarkblue,
  citecolor=mydarkblue,
  filecolor=mydarkblue,
  urlcolor=mydarkblue}
\setcitestyle{authoryear,round,citesep={;},aysep={,},yysep={;}}

\title{Competition, Alignment, and Equilibria in Digital Marketplaces}

\makeatletter
\newcommand{\printfnsymbol}[1]{%
  \textsuperscript{\@fnsymbol{#1}}%
}
\makeatother

\usepackage{authblk}
\author{Meena Jagadeesan, Michael I. Jordan, and Nika Haghtalab}
\affil{University of California, Berkeley}
\date{}

\begin{document}

\maketitle

\begin{abstract}
Competition between traditional platforms is known to improve user utility by aligning the platform's actions with user preferences. But to what extent is alignment exhibited in data-driven marketplaces? To study this question from a theoretical perspective, we introduce a duopoly market where platform actions are bandit algorithms and the two platforms compete for user participation. A salient feature of this market is that the quality of recommendations depends on both the bandit algorithm and the amount of data provided by interactions from users. This interdependency between the algorithm performance and the actions of users complicates the structure of market equilibria and their quality in terms of user utility. Our main finding is that competition in this market does not perfectly align market outcomes with user utility. Interestingly, market outcomes exhibit misalignment not only when the platforms have separate data repositories, but also when the platforms have a shared data repository. Nonetheless, the data sharing assumptions impact what mechanism drives misalignment and also affect the specific form of misalignment (e.g. the quality of the best-case and worst-case market outcomes). More broadly, our work illustrates that competition in digital marketplaces has subtle consequences for user utility that merit further investigation. 
\end{abstract}

\section{Introduction}
Recommendation systems are the backbone of numerous digital platforms---from web search engines to video sharing websites to music streaming services. To produce high-quality recommendations, these platforms rely on data which is obtained through interactions with users. This fundamentally links the quality of a platform’s services to how well the platform can attract users.

What a platform must do to attract users depends on the amount of competition in the marketplace. If the marketplace has a single platform---such as Google prior to Bing or Pandora prior to Spotify---then the platform can accumulate users by providing any reasonably acceptable quality of service given the lack of alternatives. This gives the platform great flexibility in its choice of recommendation algorithm. In contrast, the presence of competing platforms makes user participation harder to achieve and intuitively places greater constraints on the recommendation algorithms. This raises the questions: \textit{how does competition impact the recommendation algorithms chosen by digital platforms? How does competition affect the quality of service for users?}

Conventional wisdom tells us that competition benefits users. In particular, users vote with their feet by choosing the platform on which they participate. The fact that users have this power forces the platforms to fully cater to user choices and thus improves user utility. This phenomenon has been formalized in classical markets where firms produce homogenous products \citep{bertrand}, where competition has been established to perfectly align market outcomes with user utility. Since user wellbeing is considered central to the healthiness of a market, perfect competition is traditionally regarded as the “gold standard” for a healthy marketplace: this conceptual principle underlies measures of market power \citep{lerner} and antitrust policy \citep{dukelaw}.

In contrast, competition has an ambiguous relationship with user wellbeing in digital marketplaces, where digital platforms are data-driven and compete via recommendation algorithms that rely on data from user interactions. Informally speaking, these marketplaces exhibit an interdependency between user utility, the platforms' choices of recommendation algorithms, and the collective choices of other users. In particular, the size of a platform's user base impacts how much data the platform has and thus the quality of its service; as a result, an individual user's utility level depends on the number of users that the platform has attracted thus far. Having a large user base enables a platform to have an edge over competitors without fully catering to users, which casts doubt on whether classical alignment insights apply to digital marketplaces. 

The ambiguous role of competition in digital marketplaces---which falls outside the scope of our classical understanding of competition power---has gained center stage in recent policymaking discourse. Indeed, several interdisciplinary policy reports \citep{stiger19, cremer2019competition} have been dedicated to highlighting ways in which the structure of digital marketplaces fundamentally differs from that of classical markets. For example, these reports suggest that data accumulation can encourage market tipping, which leaves users particularly vulnerable to harm (as we discuss in more detail at the end of Section \ref{subsec:contributions}). Yet, no theoretical foundation has emerged to formally examine the market structure of digital marketplaces and assess potential interventions. 


\subsection{Our contributions}\label{subsec:contributions}
Our work takes a step towards building a theoretical foundation for studying competition in digital marketplaces. We present a framework for studying platforms that compete on the basis of learning algorithms, focusing on alignment with user utility at equilibrium.  We consider a stylized duopoly model based on a multi-armed bandit problem where user utility depends on the incurred rewards. We show that \textit{competition may no longer perfectly align market outcomes with user utility}. Nonetheless, we find that market outcomes exhibit \textit{a weaker form of alignment}: the user utility is at least as large as the optimal utility in a population with only one user. Interestingly, there can be multiple equilibria, and the gap between the best equilibria and the worst equilibria can be substantial.

\paragraph{Model.} 
We consider a market with two platforms and a population of users. Each platform selects a bandit algorithm from a class $\mathcal{A}$. After the platforms commit to algorithms, each user decides which platform they wish to participate on. Each user's utility is the (potentially discounted) cumulative reward that they receive from the bandit algorithm of the platform that they chose. Users arrive at a Nash equilibrium.\footnote{In Section \ref{sec:model}, we will discuss subtleties that arise from having multiple Nash equilibria.} Each platform's utility is the number of users who participate on that platform, and the platforms arrive at a Nash equilibrium. The platforms either maintain \textit{separate data repositories} about the rewards of their own users, or the platforms maintain a \textit{shared data repository} about the rewards of all users.

\paragraph{Alignment results.} 
To formally consider alignment, we introduce a metric---that we call the \textit{user quality level}---that captures the utility that a user would receive when a given pair of competing bandit algorithms are implemented and user choices form an equilibrium. Table \ref{table:main} summarizes the alignment results in the case of a single user and multiple users. A key quantity that appears in the alignment results is $R_{A'}(n)$, which denotes the expected utility that a user receives from the algorithm $A'$ when $n$ users all participate in the same algorithm.

For the case of a single user, an \textit{idealized} form of alignment holds: the user quality level at any equilibrium is the optimal utility $\max_{A'} R_{A'}(1)$ that a user can achieve within the class of algorithms $\mathcal{A}$. Idealized alignment holds regardless of the informational assumptions on the platform. 

The nature of alignment fundamentally changes when there are multiple users. At a high level, we show that idealized alignment breaks down since the user quality level is no longer guaranteed to be the global optimum, $\max_{A'} R_{A'}(N)$, that cooperative users can achieve. Nonetheless, a weaker form of alignment holds: the user quality level nonetheless never falls below the single-user optimum $\max_{A'} R_{A'}(1)$. Thus, the presence of other users cannot make a user worse off than if they were the only participant, but users may not be able to fully benefit from the data provided by others. 

More formally, consider the setting where the platforms have separate data repositories. We show that there can be many qualitatively different Nash equilibria for the platforms. The user quality level across all equilibria actually spans the full set $[\max_{A'} R_{A'}(1), \max_{A'} R_{A'}(N)]$; i.e., any user quality level is realizable in some Nash equilibrium of the platforms and its associated Nash equilibrium of the users (Theorem \ref{thm:realizable}). Moreover,  the user quality level at any equilibrium is contained in the set $[\max_{A'} R_{A'}(1), \max_{A'} R_{A'}(N)]$ (Theorem \ref{thm:utilitysepinfo}). When the number of users $N$ is large, the gap between $\max_{A'} R_{A'}(1)$ and $\max_{A'} R_{A'}(N)$ can be significant since the latter is given access to $N$ times as much data at each time step than the former. The fact that the single-user optimum $\max_{A'} R_{A'}(1)$ is realizable means that the market outcome might only exhibit a weak form of alignment. The intuition behind this result is that the performance of an algorithm is controlled not only by its efficiency in transforming information to action, but also by the level of data it has gained through its user base. Since platforms have separate data repositories, a platform can thus make up for a suboptimal algorithm by gaining a significant user base. On the other hand, the global optimal user quality level $R_{A'}(N)$ is nonetheless realizable---this suggests that equilibrium selection could be used to determine when bad equilibria arise and to nudge the marketplace towards a good equilibrium.

\textit{What if the platforms were to share data?} At first glance, it might appear that with data sharing, a platform can no longer make up for a suboptimal algorithm with data, and the idealized form of alignment would be recovered. However, we construct two-armed bandit problem instances where every symmetric equilibrium for the platforms has user quality level strictly below the global optimal $\max_{A'} R_{A'}(N)$ (Theorems \ref{thm:BHundiscounted}-\ref{thm:BH}). The mechanism for this suboptimality is that the global optimal solution requires ``too much'' exploration. If other users engage in their ``fair share'' of exploration, an individual user would prefer to explore less and free-ride off of the data obtained by other users. The platform is thus forced to explore less, which drives down the user quality level. To formalize this, we establish a connection to strategic experimentation \citep{BH98}. We further show that although all of the user quality levels in $[\max_{A'} R_{A'}(1), \max_{A'} R_{A'}(N)]$ may not be realizable, the user quality level at any symmetric equilibria is still guaranteed to be within this set (Theorem \ref{thm:utilitysharedinfo}).

\paragraph{Connection to policy reports.}  Our work provides a mathematical explanation of phenomena documented in recent policy reports \citep{stiger19, cremer2019competition}. The first phenomena that we consider is \textit{market dominance from data accumulation}. The accumulation of data  has been suggested to result in winner-takes-all-markets where a single player dominates and where market entry is challenging \citep{stiger19}. The data advantage of the dominant platform can lead to lower quality services and lower user utility. Theorems \ref{thm:realizable}-\ref{thm:utilitysepinfo} formalize this mechanism. We show that once a platform has gained the full user base, market entry is impossible and the platform only needs to achieve weak alignment with user utility to retain its user base (see discussion in Section \ref{sec:discussion}). The second phenomena that we consider is the \textit{impact of shared data access.}  While the separate data setting captures much of the status quo of proprietary data repositories in digital marketplaces, sharing data access has been proposed as a solution to market dominance \citep{cremer2019competition}. Will shared data access deliver on its promises? Theorems \ref{thm:BHundiscounted}-\ref{thm:BH} highlight that sharing data does not solve the alignment issues, and uncovers free-riding as a mechanism for misalignment.

\begin{table}[]
    \centering
    \begin{tabular}{ccc}\toprule
     &  \textbf{Single user}  & \textbf{Multiple users}   \\ 
         \cmidrule(lr){2-2}\cmidrule(lr){3-3}\\[-3mm]
 \textbf{Separate data repositories}   & $\max_{A'} R_{A'}(1)$   & $[\max_{A'} R_{A'}(1), \max_{A'} R_{A'}(N)]$   \\ 
           & & \\ [-1mm]
  \multirow{2}{*}{\textbf{Shared data repository}}   &  \multirow{2}{*}{$\max_{A'} R_{A'}(1)$}  & subset of $[\max_{A'} R_{A'}(1), \max_{A'} R_{A'}(N)]$ \\ 
  & & (\textit{strict} subset in safe-risky arm problem)\\
    \bottomrule\\[-2mm]
    \end{tabular}
    \caption{User quality level of the Nash equilibrium for the platforms. A marketplace with a single user exhibits \textit{idealized alignment}, where the user quality level is maximized. A marketplace with multiple users can have equilibria with a vast range of user quality levels---although \textit{weak alignment} always holds----and there are subtle differences between the separate and shared data settings.}
    \label{table:main}
\end{table}

\subsection{Related work}

We discuss the relation between our work and research on \textit{competing platforms}, \textit{incentivizing exploration}, and \textit{strategic experimentation}.

\paragraph{Competing platforms.} \citet{AMSW20} examine the interplay between competition and exploration in bandit problems in a duopoly economy with fully myopic users. They focus on platform regret, showing that platforms must both choose a greedy algorithm at equilibrium and thus illustrating that competition is at odds with regret minimization. In contrast, we take a user-centric perspective and investigate alignment with user utility. Interestingly, the findings in \citet{AMSW20} and our findings are not at odds: the result in \citet{AMSW20} can be viewed as alignment (since the optimal choice for a fully myopic user results in regret in the long run), and our results similarly recover idealized alignment in the special case when users are fully myopic. Going beyond the setup of \citet{AMSW20}, we investigate non-myopic users and allow multiple users to arrive at every round, and we show that alignment breaks down in this general setting. 

Outside of the bandits framework, another line of work has also studied the behavior of competing learners when users can choose between platforms. \citet{BT17, BT19} study equilibrium predictors chosen by competing \textit{offline} learners in a PAC learning setup. Other work has focused on the dynamics when multiple learners apply out-of-box algorithms, showing that specialization can emerge \citep{GZKZ21, DCRMF22} and examining the role of data purchase \citep{KGZ22}; however, these works do not consider which algorithms the learners are incentivized to choose to gain users. In contrast, we investigate \textit{equilibrium} bandit algorithms chosen by \textit{online} learners, each of whom aims to maximize the size of its user base. The interdependency between the platforms' choices of algorithms, the data available to the platforms, and the  users' decisions in our model drives our alignment insights.

Other aspects of competing platforms that have been studied include competition under exogeneous network effects \citep{R09, WW14}, experimentation in price competition \citep{BV2000}, dueling algorithms which compete for a single user \citep{IKLMPT11}, competition when firms select scalar innovation levels whose cost decreases with access to more data \citep{PS21}, and measures of a digital platform's power in a marketplace \citep{HJM22}.

\paragraph{Incentivizing exploration.} This line of work has examined how the availability of outside options impacts bandit algorithms. \citet{KMP13} show that Bayesian Incentive Compatibility (BIC) suffices to guarantee that users will stay on the platform. Follow-up work (e.g., \cite{MSS15, SS21}) examines what bandit algorithms are BIC. \citet{FKKK14} explore the use of monetary transfers. 

\paragraph{Strategic experimentation.} This line of work has investigated equilibria when a population of users each choose a bandit algorithm.  \citet{BH98, BH00, BHSimple} analyze the equilibria in a risky-safe arm bandit problem: we leverage their results in our analysis of equilibria in the shared data setting. Strategic experimentation (see \cite{HS17} for a survey) has investigated exponential bandit problems \citep{KRC15}, the impact of observing actions instead of payoffs \citep{RSV07}, and the impact of cooperation \citep{BP21}.

\section{Model}\label{sec:model}
We consider a duopoly market with two platforms performing a multi-armed bandit learning problem and a population of $N$ users, $u_1, \ldots, u_N$, who choose between platforms. Platforms commit to bandit algorithms, and then each user chooses a single platform to participate on for the learning task.

\subsection{Multi-armed bandit setting}\label{sec:bandit} 

Consider a Bayesian bandit setting where there are $k$ arms with priors $\DPrior_1, \ldots, \DPrior_k$. At the beginning of the game, the mean rewards of arms are drawn from the priors $r_1 \sim \DPrior_1, \ldots, r_k \sim \DPrior_k$. These mean rewards are unknown to both the users and the platforms but are shared across the two platforms. If the user's chosen platform recommends arm $i$, the user receives reward drawn from a noisy distribution $\DNoise(r_i)$ with mean $r_i$. 

Let $\mathcal{A}$ be a class of bandit algorithms that  map the information state given by the posterior distributions to an arm to be pulled. The \textit{information state} $\IS = [\DPosterior_1, \ldots, \DPosterior_k]$ is taken to be the set of posterior distributions for the mean rewards of each arm. We assume that each algorithm $A \in \mathcal{A}$ can be expressed as a function mapping the information state $\IS$ to a distribution over arms in $[k]$.\footnote{This assumption means that an algorithm's choice is independent of the time step conditioned on $\IS$. Classical bandit algorithms such as Thompson sampling \citep{T33}, finite-horizon UCB \citep{LR85}, and the infinite-time Gittins index \citep{G79} fit into this framework. This assumption is \textit{not} satisfied by the infinite time horizon UCB.} We let $A(\IS)$ denote this distribution over arms $[k]$.

\paragraph{Running example: risky-safe arm bandit problem.} To concretize our results, we consider the \textit{risky-safe arm bandit problem} as a running example. The noise distribution $\DNoise(r_i)$ is a Gaussian $N(r_i,\sigma^2)$. The first arm is a \textit{risky arm} whose prior distribution $\DPrior_1$ is over the set $\left\{l, h\right\}$, where $l$ corresponds to a ``low reward'' and $h$ corresponds to a ``high reward.'' The second arm is a \textit{safe arm} with known reward $s \in (l,h)$ (the prior $\DPrior_2$ is a point mass at $s$). In this case, the information state $\IS$ permits a one-dimensional representation given by the posterior probability $p(\IS) := \mathbb{P}_{X \sim \DPosterior_1} [X = h]$ that the risky arm is high reward. 

We construct a natural algorithm class as follows. For a measurable function $f:[0,1] \rightarrow [0,1]$, let $A_f$ be the associated algorithm defined so $A_f(\IS)$ is a distribution that is 1 with probability $f(p(\IS))$ and 2 with probability $1 - f(p(\IS))$. We define 
\[\AAll := \left\{A_f \mid f: [0,1] \rightarrow [0,1] \text{ is measurable} \right\}\]
to be the class of all randomized algorithms. This class contains Thompson sampling ($A_{f_{\textrm{TS}}}$ is given by $f_{\textrm{TS}}(p) = p$), the Greedy algorithm ($A_{f_{\textrm{Greedy}}}$ is given by $f_{\textrm{Greedy}}(p) = 1$ if $ph + (1-p)l \ge s$ and $f_{\textrm{Greedy}}(p) = 0$ otherwise), and mixtures of these algorithms with uniform exploration. We consider restrictions of the class $\AAll$ in some results.

\subsection{Interactions between platforms, users, and data}\label{sec:actions}
The interactions between the platform and users impact the data that the platform receives for its learning task. The platform action space $\mathcal{A}$ is a class of bandit algorithms that map an information state $\mathcal{I}$ to an arm to be pulled. The user action space is $\left\{1, 2\right\}$. For $1 \le i \le N$, we denote by $p^i \in \left\{1, 2 \right\}$ the action chosen by user $u_i$. 

\paragraph{Order of play.} The platforms commit to algorithms $A_1$ and $A_2$ respectively, and then users simultaneously choose their actions $p^1, \ldots, p^N$ prior to the beginning of the learning task. We emphasize that user $i$ participates on platform $p^i$ for the \textit{full duration} of the learning task. (In Appendix \ref{sec:users}, we discuss the assumption that users cannot switch platforms between time steps.)

\paragraph{Data sharing assumptions.} In the \textit{separate data repositories} setting, each platform has its own (proprietary) data repository for keeping track of the rewards incurred by its own users. Platforms 1 and 2 thus have separate information states given by $\IS_1 = [\DPosterior_{1, 1}, \ldots, \DPosterior_{1, k}]$ and $\IS_2 = [\DPosterior_{2, 1}, \ldots, \DPosterior_{2, k}]$, respectively. In the \textit{shared data repository} setting, the platforms share an information state $\ISShared = [\DPosterior_1, \ldots, \DPosterior_k]$, which is updated based on the rewards incurred by users of both platforms.\footnote{In web search, recommender systems can query each other, effectively building a shared information state.}

\paragraph{Learning task.} The learning task is determined by the choice of platform actions $A_1$ and $A_2$, user actions $p^1, \ldots, p^n$, and specifics of data sharing between platforms. At each time step: 
\begin{enumerate}
    \item Each user $u_i$ arrives at platform $p^i$. The platform $p^i$ recommends arm $a_i \sim A_i(\IS)$ to that user, where $\IS$ denotes the information state of the platform. (The randomness of arm selection is fully independent across users and time steps.) The user $u_i$ receives noisy reward $\DNoise(r_{a_i})$.
    \item After providing recommendations to all of its users, platform 1 observes the rewards incurred by users in $S_1 := \left\{i \in [N] \mid p^i = 1\right\}$. Platform 2 similarly observes the rewards incurred by users in $S_2 := \left\{i \in [N] \mid p^i = 2 \right\}$. Each platform then updates their information state $\IS$ with the corresponding posterior updates.  
    \item A platform may have access to external data that does not come from users. To capture this, we introduce background information into the model. Both platforms observe the same \textit{background information} of quality $\sigma_b \in (0, \infty]$. In particular, for each arm $i$, the platforms observe the same realization of a noisy reward $\DNoise(r_{i})$. When $\sigma_b = \infty$, we say that there is \textit{no background information} since the background information is uninformative. The corresponding posterior updates are then used to update the information state ($\mathcal{I}$ in the case of shared data; $\mathcal{I}_1$ and $\mathcal{I}_2$ in the case of separate data). 
\end{enumerate}
In other words, platforms receive information from users (and background information), 
and users receive rewards based on the recommendations of the platform that they have chosen.

\subsection{Utility functions and equilibrium concept}\label{sec:utility}
User utility is generated by rewards, while the platform utility is generated by \textit{user participation}. 

\paragraph{User utility function.} We follow the standard discounted formulation for bandit problems (e.g. \citep{GJ79, BH98}), where the utility incurred by a user is defined by the expected (discounted) cumulative reward received across time steps. The discount factor $\beta$ parameterizes the extent to which agents are myopic. Let $\U(p^i; \vectorfont{p}^{-i}, A_1, A_2)$ denote the utility of a user $u_i$ if they take action $p^i \in \left\{1,2\right\}$ when other users take actions $\vectorfont{p}^{-i} \in \left\{1,2\right\}^{N-1}$ and the platforms choose $A_1$ and $A_2$. For clarity, we make this explicit in the case of discrete time setup with horizon length $T \in [1, \infty]$. Let $a_i^t = a_i^t(A_1, A_2, \vectorfont{p})$ denote the arm recommended to user $u_i$ at time step $t$. The utility  is defined to be
\[\U(p^i; \vectorfont{p}^{-i}, A_1, A_2) := \mathbb{E}\left[\sum_{t=1}^T \beta^t r_{a_i^t}\right]\]
where the expectation is over randomness of the incurred rewards and the algorithms.
In the case of continuous time, the utility is 
\[\U(p^i; \vectorfont{p}^{-i}, A_1, A_2) := \mathbb{E}\left[\int e^{-\beta t} d\pi(t)\right] \]
where the $\beta \in [0, \infty)$ denotes the discount factor and $d\pi(t)$ denotes the payoff received by the user.\footnote{For discounted utility, it is often standard to introduce a multiplier of $\beta$ for normalization (see e.g. \citep{BH98}). The utility $\U(p^i; \vectorfont{p}^{-i}, A_1, A_2)$ could have equivalently be defined as $\mathbb{E}\left[\int \beta e^{-\beta t} d\pi(t)\right]$ without changing any of our results.} In both cases, observe that the utility function is \textit{symmetric} in user actions.

The utility function implicitly differs in the separate and shared data settings, since the information state evolves differently in these two settings. When we wish to make this distinction explicit, we denote the corresponding utility functions by $\USeparate$ and $\UShared$.

\paragraph{User equilibrium concept.} 
We assume that after the platforms commit to algorithms $A_1$ and $A_2$, the users end up at a pure strategy Nash equilibrium of the resulting game. More formally,  let $\vectorfont{p} \in \left\{1,2\right\}^N$ be a \textit{pure strategy Nash equilibrium for the users} if $p_i \in \argmax_{p \in \left\{0,1\right\}} \U(p; \vectorfont{p}^{-i}, A_1, A_2)$ for all $1 \le i \le N$. The existence of a pure strategy Nash equilibrium follows from the assumption that the game is symmetric and the action space has 2 elements \citep{C04}.

One subtlety is that there can be multiple equilibria in this general-sum game. For example, there are always at least 2 (pure strategy) equilibria when platforms play any $(A, A)$, i.e., commit to the same algorithm --- one equilibrium where all users choose the first platform, and another where all users choose the second platform). Interestingly, there can be multiple equilibria even when one platform chooses a ``worse'' algorithm than the other platform. 
We denote by $\mathcal{E}_{A_1, A_2}$ the set of pure strategy Nash equilibria when the platforms choose algorithms $A_1$ and $A_2$. We simplify the notation and use $\mathcal{E}$ when $A_1$ and $A_2$ are clear from the context. In Section \ref{sec:mixed}, we discuss our choice of solution concept, focusing on what the implications would have been of including mixed Nash equilibria in $\mathcal{E}$.

\paragraph{Platform utility and equilibrium concept.} The utility of the platform roughly corresponds to the number of users who participate on that platform. This captures that in markets for digital goods, where platform revenue is often derived from advertisement or subscription fees, the number of users serviced is a proxy for platform revenue. 

When formalizing the utility that a platform receives, the fact that there can be several user equilibria for a given choice of platform algorithms creates ambiguity. To resolve this, we consider the worst-case user equilibrium for the platform, and we define platform utility to be the \textit{minimum number of users that a platform would receive at any pure strategy equilibrium for the users.}
More formally, when platform 1 chooses algorithm $A_1$ and platform 2 chooses algorithm $A_2$, the utilities of platform 1 and platform 2 are given by: 
\begin{equation}
\label{eq:platformutility}
v_1(A_1; A_2) := \min_{\vectorfont{p} \in \mathcal{E}} \sum_{i=1}^N \Indicator[p^i = 1] \;\;\;\ \text{      and     }  \;\;\;\  v_2(A_2; A_1) = \min_{\vectorfont{p} \in \mathcal{E}} \sum_{i=1}^N \Indicator[p^i = 2].  
\end{equation}
The minimum over equilibria $\vectorfont{p} \in \mathcal{E}$ in \eqref{eq:platformutility} captures a worst-case perspective where a platform makes decisions based on the worst possible utility that they might receive by choosing a given algorithm.\footnote{Interestingly, if we were to define the platform utility in \eqref{eq:platformutility} to be a maximum over equilibria $\vectorfont{p} \in \mathcal{E}$, this would induce degenerate behavior: any symmetric solution $(A, A)$ would maximize platform utility and thus be an equilibrium. In contrast, formalizing the platform utility as the minimum over equilibria avoids this degeneracy.} 

The equilibrium concept for the platforms is a \textit{pure strategy Nash equilibrium}, and we often focus on \textit{symmetric} equilibria. We discuss the existence of such an equilibrium in Sections \ref{sec:separate}-\ref{sec:shared}. We note that at equilibrium, the utility for the platforms is typically $0$, aligning with classical economic intuition. In particular, platforms earning zero equilibrium utility in our model mirrors firms earning zero equilibrium profit in price competition \citep{bertrand}. However, there is an important distinction: platform utility \textit{ex-post} (after users choose between platforms) may no longer be $0$ and in fact may be as large as $N$, while firm profit in price competition remains $0$ ex-post.

\section{Formalizing the Alignment of a Market Outcome} 

The \textit{alignment} of an equilibrium outcome for the platforms is measured by the amount of user utility that it generates. In Section \ref{sec:uql} we introduce the \textit{user quality level} to formalize alignment. In Section \ref{sec:idealized}, we show an idealized form of alignment for $N = 1$ (Theorem \ref{thm:idealized}). In Section \ref{sec:benchmarks}, we turn to the case of multiple users and discuss benchmarks for the user quality level. In Section \ref{subsec:assumptions}, we describe mild assumptions on $\mathcal{A}$ that we use in our alignment results for multiple users.

\subsection{User quality level} \label{sec:uql}
Given a pair of platform algorithms $A_1 \in \mathcal{A}$ and $A_2 \in \mathcal{A}$, we introduce the \textit{user quality level} $Q(A_1, A_2)$ to measure the alignment between platform algorithms and user utility. Informally speaking, the user quality level $Q(A_1, A_2)$ captures the utility that a user would receive when the platforms choose algorithms $A_1$ and $A_2$ and when user choices form an equilibrium.

When formalizing the user quality level, the potential multiplicity of user equilibria creates ambiguity (like in the definition of platform utility in \eqref{eq:platformutility}), and different users potentially receiving different utilities creates further ambiguity. We again take a worst-case perspective and formalize the user quality level as the minimum over equilibria $\vectorfont{p} \in \mathcal{E}$ and over users $1 \le i \le N$. 
\begin{definition}[User quality level]
\label{def:userqualitylevel}
Given algorithms $A_1$ and $A_2$ chosen by the platforms, the user quality level is defined to be $Q(A_1, A_2) := \min_{\vectorfont{p} \in \mathcal{E}, 1 \le i \le N} \U(p_i; \vectorfont{p}^{-i}, A_1, A_2)$. 
\end{definition}
\noindent Interestingly, our insights about alignment would remain unchanged if we were to define the user quality level based on an arbitrary user equilibrium and user, rather than taking a minimum. More specifically, our alignment results (Theorems \ref{thm:realizable}, \ref{thm:utilitysepinfo}, \ref{thm:BHundiscounted}, \ref{thm:BH}, \ref{thm:utilitysharedinfo}) would still hold if $Q(A_1, A_2)$ were defined to be $\U(p_i; \vectorfont{p}^{-i}, A_1, A_2)$ for any arbitrarily chosen $\vectorfont{p} \in \mathcal{E}$ and $1 \le i \le N$.\footnote{For most of results (Theorems \ref{thm:realizable},  \ref{thm:BHundiscounted}, \ref{thm:BH}, \ref{thm:utilitysharedinfo}), the reason that the results remain unchanged is that at a symmetric solution $(A, A)$, the user utility turns out to be the same for all $\vectorfont{p} \in \mathcal{E}$ and $1 \le i \le N$ (this follows by definition for the shared data setting and follows from Lemma \ref{lemma:pure} for the separate data setting). 
The reason that Theorem \ref{thm:utilitysepinfo}, which considers asymmetric solutions, remains unchanged is that the lower bound holds for the worst-case $\vectorfont{p} \in \mathcal{E}$ and $1 \le i \le N$ (and thus for any $\vectorfont{p} \in \mathcal{E}$ and $1 \le i \le N$) and the proof of the upper bound applies more generally to any selection of $\vectorfont{p} \in \left\{1,2\right\}^N$ and $1 \le i \le N$.} This demonstrates that our alignment results are independent of the particularities of how the user quality level is formalized.

To simplify notation in our alignment results, we introduce the \textit{reward function} which captures how the utility that a given algorithm generates changes with the number of users who contribute to its data repository. For an algorithm $A \in \mathcal{A}$, let the reward function $R_A: [N] \rightarrow \mathbb{R}$ be defined by:
\[R_A(n) := \USeparate(1; \vectorfont{p}_{n-1}, A, A) ,\]
where $\vectorfont{p}_{n-1}$ corresponds to a vector with $n-1$ coordinates equal to one.

\subsection{Idealized alignment result: The case of a single user}\label{sec:idealized}

When there is a single user, the platform algorithms turn out to be perfectly aligned with user utilities at equilibrium. To formalize this, we consider the optimal utility that could be obtained by a user across any choice of actions by the platforms and users (not necessarily at equilibrium): that is, $ \max_{p \in \left\{1,2\right\}, A_1 \in \mathcal{A}, A_2 \in \mathcal{A}}\U(p; \emptyset, A_1, A_2)$. Using the setup of the single-user game, we can see that this is equal to  $\max_{A \in \mathcal{A}}\U(1; \emptyset, A, A) = \max_{A \in \mathcal{A}} R_A(1)$. We show that the user quality level always meets this benchmark (we defer the proof to Appendix \ref{appendix:proofidealized}). 
\begin{theorem}
\label{thm:idealized}
Suppose that $N = 1$, and consider either the separate data setting or the shared data setting. If $(A_1, A_2)$ is a pure strategy Nash equilibrium for the platforms, then the user quality level $Q(A_1, A_2)$ is equal to $\max_{A \in \mathcal{A}}R_A(1)$.
\end{theorem}

Theorem \ref{thm:idealized} shows that in a single-user market, two firms is sufficient to perfectly align firm actions with user utility---this stands in parallel to 
classical Bertrand competition in the pricing setting \citep{bertrand}. 

\paragraph{Proof sketch of Theorem \ref{thm:idealized}.}There are only 2 possible pure strategy equilibria: either the user chooses platform 1 and receives utility $R_{A_1}(1)$ or the user chooses platform 2 and receives utility $R_{A_2}(1)$. If one platform chooses a suboptimal algorithm for the user (i.e. an algorithm $A'$ where $R_{A'}(1) < \max_{A \in \mathcal{A}} R_A(1)$), then the other platform will receive the user (and thus achieve utility 1) if they choose a optimal algorithm $\argmax_{p \in \left\{1, 2\right\}} R_{A_p}(1)$. This means that $(A_1, A_2)$ is a pure strategy Nash equilibrium if and only if $A_1 \in \argmax_{A' \in \mathcal{A}} R_{A'}(1)$ or $A_2 \in \argmax_{A' \in \mathcal{A}} R_{A'}(1)$. The user thus receives utility $\max_{A \in \mathcal{A}} R_{A'}(1)$. We defer the full proof to Appendix \ref{appendix:proofidealized}.

\subsection{Benchmarks for user quality level} \label{sec:benchmarks}

In the case of multiple users, this idealized form of alignment turns out to break down, and formalizing alignment requires a more nuanced consideration of benchmarks. We define the \textit{single-user optimal utility} of $\mathcal{A}$ to be $\max_{A \in \mathcal{A}} R_A(1)$. This corresponds to maximal possible user utility that can be generated by a platform who only serves a single user and thus relies on this user for all of its data.  On the other hand, we define the \textit{global optimal utility} of $\mathcal{A}$ to be $\max_{A \in \mathcal{A}} R_A(N)$. This corresponds to the maximal possible user utility that can be generated by a platform when all of the users in the population are forced to participate on the same platform. The platform can thus maximally enrich its data repository in each time step.

\subsection{Assumptions on $\mathcal{A}$}\label{subsec:assumptions}

While our alignment results for a single user applied to arbitrary algorithm classes, we require mild assumptions on $\mathcal{A}$ in the case of multiple users to endow the equilibria with basic structure. 

\textit{Information monotonicity} requires that an algorithm $A$'s performance in terms of user utility does not worsen with additional posterior updates to the information state. Our first two instantations of information monotonicity---strict information monotonicity and information constantness---require that the user utility of $A$ grow monotonically in the number of other users participating in the algorithm. Our third instantation of information monotonicity---side information monotonicity---requires that the user utility of $A$ not decrease if other users also update the information state, regardless of what algorithm is used by the other users. We formalize these assumptions as follows: 
\begin{assumption}[Information monotonicity]
\label{assumption:IM} 
For any given discount factor $\beta$ and number of users $N$, an algorithm $A \in \mathcal{A}$ is \textbf{strictly information monotonic} if $R_A(n)$ is strictly increasing in $n$ for $1 \le n \le N$. An algorithm $A$ is \textbf{information constant} if $R_A(n)$ is constant in $n$ for $1 \le n \le N$. An algorithm $A$ is \textbf{side information monotonic} if for every measurable function $f$ mapping information states to distributions over $[k]$ and for every $1 \le n \le N - 1$, it holds that $\UShared(1; \mathbf{2}_{n}, A, f) \ge R_A(1)$ where $\mathbf{2}_{n} \in \left\{1,2\right\}^n$ has all coordinates equal to $2$.
\end{assumption}

While information monotonicity places assumptions on each algorithm in $\mathcal{A}$, our next assumption places a mild restriction on how the utilities generated by algorithms in $\mathcal{A}$ relate to each other. \textit{Utility richness} requires that the set of user utilities spanned by $\mathcal{A}$ is a sufficiently rich interval.
\begin{assumption}[Utility richness]
\label{assumption:richness}
A class of algorithms $\mathcal{A}$ is \textbf{utility rich} if the set of utilities $\left\{ R_A(N) \mid A \in \mathcal{A} \right\}$ is a contiguous set, the supremum of $\left\{ R_A(N) \mid A \in \mathcal{A} \right\}$ is achieved, and  there exists $A' \in \mathcal{A}$ such that $R_{A'}(N) \le \max_{A \in \mathcal{A}} R_A(1)$. 
\end{assumption}

\paragraph{Discussion of assumptions.} Utility richness holds (almost) without loss of generality, by taking the closure of an algorithmic class under the operation of mixing with the pure exploration algorithm (see Lemma \ref{lemma:UR}). On the other hand, not all algorithms are information monotone. Nevertheless, we show that information monotonicity is satisfied for several algorithms for the risky-safe arm setup, including any nondegenerate algorithm under undiscounted rewards (see Lemma \ref{lemma:undiscountedIM}) and Thompson sampling under discounted rewards (see Lemma \ref{lemma:TS}). These results are of independent interest, and more broadly, understanding information monotonicity is crucial for investigating the incentive properties of bandit algorithms: indeed prior work (e.g. \citet{AMSW20}) has explored variants of this assumption. We defer a detailed discussion of these assumptions to Section \ref{sec:assumptions}. 

\section{Separate data repositories}\label{sec:separate}

We investigate alignment when the platforms have separate data repositories. In Section \ref{sec:alignmentseparate}, we show that there can be many qualitatively different equilibria for the platforms and characterize the alignment of these equilibria. In Section \ref{sec:discussion}, we discuss factors that drive the level of misalignment in a marketplace. 

\subsection{Multitude of equilibria and the extent of alignment}\label{sec:alignmentseparate}

In contrast with the single user setting, the marketplace can exhibit multiple equilibria for the platforms. As a result, to investigate alignment, we investigate the \textit{range} of achievable user quality levels. Our main finding is that the equilibria in a given marketplace can exhibit a vast range of alignment properties. In particular, \textit{every} user quality level in between the single-user optimal utility $\max_{A' \in \mathcal{A}} R_{A'}(1)$ and the global optimal utility $\max_{A' \in \mathcal{A}} R_{A'}(N)$ can be realized by some equilibrium for the platforms. 
\begin{theorem}
\label{thm:realizable}
Suppose that each algorithm in $\mathcal{A}$ is either strictly information monotonic or information constant (Assumption \ref{assumption:IM}), and suppose that $\mathcal{A}$ is utility rich (Assumption \ref{assumption:richness}). For every $\alpha \in [\max_{A' \in \mathcal{A}} R_{A'}(1), \max_{A' \in \mathcal{A}} R_{A'}(N)]$, there exists a symmetric pure strategy Nash equilibrium $(A, A)$ in the separate data setting such that $Q(A, A) = \alpha$.
\end{theorem}
Nonetheless, there is a baseline (although somewhat weak) form of alignment achieved by all equilibria. In particular, every equilibrium for the platforms has user quality level at least the single-user optimum $\max_{A' \in \mathcal{A}} R_{A'}(1)$. 
\begin{theorem}
\label{thm:utilitysepinfo}
Suppose that each algorithm in $\mathcal{A}$ is either strictly information monotonic or information constant (see Assumption \ref{assumption:IM}). In the separate data setting, at any pure strategy Nash equilibrium $(A_1, A_2)$ for the platforms, the user quality level lies in the following interval:
\[Q(A_1, A_2) \in \left[\max_{A' \in \mathcal{A}} R_{A'}(1), \max_{A' \in \mathcal{A}} R_{A'}(N)\right].\] 
\end{theorem}

An intuition for these results is that the performance of an algorithm depends not only on how it transforms information to actions, but also on the amount of information to which it has access. A platform can make up for a suboptimal algorithm by attracting a significant user base: if a platform starts with the full user base, it is possible that \textit{no single user} will switch to the competing platform, even if the competing platform chooses a stricter better algorithm. However, if a platform's algorithm is highly suboptimal, then the competing platform will indeed be able to win the full user base.

\paragraph{Proof sketches of Theorem \ref{thm:realizable} and Theorem \ref{thm:utilitysepinfo}.}
The key idea is that pure strategy equilibria for users take a simple form. Under strict information monotonicity, we show that every pure strategy equilibrium $p^* \in \mathcal{E}^{A_1, A_2}$ is in the set $\left\{[1, \ldots, 1], [2, \ldots, 2]\right\}$ (Lemma \ref{lemma:pure}). The intuition is that the user utility strictly grows with the amount of data that the platform has, which in turn grows with the number of other users participating on the same platform. It is often better for a user to switch to the platform with more users, which drives all users to a single platform in equilibrium. 

The reward functions $R_{A_1}(\cdot)$ and $R_{A_2}(\cdot)$ determine which of these two solutions are in $\mathcal{E}^{A_1, A_2}$. It follows from definition that $[1, \ldots, 1]$ is in $\mathcal{E}_{A_1, A_2}$ if and only if $R_{A_1}(N) \ge R_{A_2}(1)$. This inequality can hold even if $A_2$ is a better algorithm in the sense that $R_{A_2}(n) > R_{A_1}(n)$ for all $n$. The intuition is that the performance of an algorithm is controlled not only by its efficiency in choosing the possible action from the information state, but also by the size of its user base. The platform with the worse algorithm can be better for users if it has accrued enough users. 

This characterization of the set $\mathcal{E}_{A_1, A_2}$ enables us to reason about the platform equilibria. To prove Theorem \ref{thm:realizable}, we show that $(A, A)$ is an equilibrium for the platforms as long as $R_A(N) \ge \max_{A'} R_{A'}(1)$. This, coupled with utility richness, enables us to show that every utility level in $[\max_{A' \in \mathcal{A}} R_{A'}(1), \max_{A' \in \mathcal{A}} R_{A'}(N)]$ can be realized. To prove Theorem \ref{thm:utilitysepinfo}, we first show platforms can't both choose highly suboptimal algorithms: in particular, if $R_{A_1}(N)$ and $R_{A_2}(N)$ are both below the single-user optimal $\max_{A' \in \mathcal{A}} R_{A'}(1)$, then $(A_1, A_2)$ is not in equilibrium. Moreover, if one of the platforms chooses an algorithm $A$ where $R_A(N) < \max_{A' \in \mathcal{A}} R_{A'}(1)$, then all of the users will choose the other platform in equilibrium. The full proofs are deferred to  Appendix \ref{sec:proofsseparate}.

\subsection{What drives the level of misalignment in a marketplace?}\label{sec:discussion}

The existence of multiple equilibria makes it more subtle to reason about the alignment exhibited by a marketplace. The level of misalignment depends on two factors: first, the size of the range of realizable user quality levels, and second, the selection of equilibrium within this range. We explore each of these factors in greater detail. 

\paragraph{How large is the range of possible user quality levels?} Both the algorithm class and the structure of the user utility function determine the size of the range of possible user quality levels. We informally examine the role of the user's discount factor on the size of this range. 

First, consider the case where users are fully non-myopic (so their rewards are undiscounted across time steps). 
The gap between the single-user optimal utility $\max_{A' \in \mathcal{A}} R_{A'}(1)$ and global optimal utility $\max_{A' \in \mathcal{A}} R_{A'}(N)$ can be substantial. To gain intuition for this, observe that the utility level $R_{A'}(N)$ corresponds to the algorithm $A'$ receiving $N$ times as much as data at every time step than the utility level $R_{A'}(1)$. For example, consider an algorithm $A'$ whose regret grows according to $\sqrt{T}$ where $T$ is the number of samples collected, and let $\textrm{OPT} := \mathbb{E}_{r_1 \sim \DPrior_1, \ldots r_k \sim \DPrior_k}[\max_{1 \le i \le k} r_i]$ be the expected maximum reward of any arm. Since utility and regret are related up to additive factors for fully non-myopic users, then we have that $R_{A'}(1) \approx \textrm{OPT} - \sqrt{T}$ while $R_{A'}(N) \approx \textrm{OPT} - \sqrt{N T}$.  

At the other extreme, consider the case where users are fully myopic. In this case, the range collapses to a \textit{single point}. The intuition is that the algorithm generates the same utility for a user regardless of the number of other users who participate: in particular, $R_{A'}(1)$ is equal to $R_{A'}(N)$ for any algorithm $A' \in \mathcal{A}$. To see this, we observe that the algorithm's behavior beyond the first time step does not factor into user utility, and the algorithm's selection at the first time is determined before it receives any information from users. Put differently, although $R_{A'}(N)$ can receives $N$ times more information, there is a delay before the algorithm sees this information. Thus, in the case of fully myopic users, the user quality level is always equal to the global optimal user utility $\max_{A} R_A(N)$  so idealized alignment is actually recovered. When users are partially non-myopic, the range is no longer a single point, but the range is intuitively smaller than in the undiscounted case.

\paragraph{Which equilibrium arises in a marketplace?.} When the gap between the single-user optimal and global optimal utility levels is substantial, it becomes ambiguous what user quality level will be realized in a given marketplace. Which equilibria arises in a marketplace depends on several factors. 

One factor is the secondary aspects of the platform objective that aren’t fully captured by the number of users. For example, suppose that the platform cares about the its reputation and thus is incentivized to optimize for the quality of the service. This could drive the marketplace towards higher user quality levels. On the other hand, suppose that the platform derives other sources of revenue from recommending certain types of content (e.g. from recommending advertisements). If these additional sources of revenue are not aligned with user utility, then this could drive the marketplace towards lower user quality levels. 

Another factor is the mechanism under which platforms arrive at equilibrium solutions, such as \textit{market entry}. We informally show that market entry can result in the the \textit{worst possible user utility} within the range of realizable levels. To see this, notice that when one platform enters the marketplace shortly before another platform, all of the users will initially choose the first platform. The second platform will win over users only if $R_{A_2}(1) > R_{A_1}(N)$, where $A_2$ denotes the algorithm of the second platform and $A_1$ denotes the algorithm of the first platform. In particular, the platform is susceptible to losing users only if $R_{A_1}(N) < \max_{A' \in \mathcal{A}} R_{A'}(1)$. Thus, the worst possible equilibrium can arise in the marketplace, and this problem only worsens if the first platform enters early enough to accumulate data beforehand. This finding provides a mathematical backing for the barriers to entry in digital marketplaces that are documented in policy reports \citep{stiger19}. 

This finding points to an interesting direction for future work: \textit{what equilibria arise from other natural mechanisms?}

\section{Shared data repository}\label{sec:shared}

What happens when data is shared between the platforms? We show that both the nature of alignment and the forces that drive misalignment fundamentally change. In Section \ref{sec:counterexample}, we show a construction where the user quality levels do not span the full set $[\max_{A'} R_{A'}(1), \max_{A'} R_{A'}(N)]$. Despite this, in Section \ref{sec:alignmentshared}, we establish that the user quality level at any symmetric equilibrium continues to be at least $\max_{A'} R_{A'}(1)$.

\subsection{Construction where global optimal is not realizable}\label{sec:counterexample} 

In contrast with the separate data setting, the set of user quality levels at symmetric equilibria for the platforms does not necessarily span the full set $[\max_{A'} R_{A'}(1), \max_{A'} R_{A'}(N)]$. To demonstrate this, we show that in the risky-safe arm problem, every symmetric equilibrium $(A,A)$ has user quality level $Q(A, A)$ strictly below $\max_{A'} R_{A'}(N)$.
\begin{theorem}
\label{thm:BHundiscounted}
Let the algorithm class $\AAllContinuous \subseteq \AAll$ consist of the algorithms $A_f$ where $f(0) = 0$, $f(1) = 1$, and $f$ is continuous  at $0$ and $1$. In the shared data setting, for any choice of prior $p \in (0,1)$ and any background information quality $\sigma_b \in (0, \infty)$, there exists an undiscounted risky-safe arm bandit setup (see Setup \ref{setup:continuoustimeundiscounted}) such that the set of realizable user quality levels for algorithm class $\AAllContinuous$ is equal to a singleton set:  
\[\left\{Q(A, A) \mid (A,A) \text{ is a symmetric equilibrium for the platforms }\right\} = \left\{\alpha^*\right\}\]
where 
\[\max_{A' \in \mathcal{A}} R_{A'}(1) < \alpha^* < \max_{A' \in \mathcal{A}} R_{A'}(N).\]
\end{theorem}
\begin{theorem}
\label{thm:BH}
In the shared data setting, for any discount factor $\beta \in (0, \infty)$ and any choice of prior $p \in (0,1)$, there exists a discounted risky-safe arm bandit setup with no background information (see Setup \ref{setup:continuoustimediscounted}) such that the set of realizable user quality levels for algorithm class $\AAll$  is equal to a singleton set:  
\[\left\{Q(A, A) \mid (A,A) \text{ is a symmetric equilibrium for the platforms }\right\} = \left\{\alpha^*\right\}\]
where 
\[\max_{A' \in \mathcal{A}} R_{A'}(1) \le \alpha^* < \max_{A' \in \mathcal{A}} R_{A'}(N).\]
\end{theorem}

Theorems \ref{thm:BHundiscounted} and \ref{thm:BH} illustrate examples where there is no symmetric equilibrium for the platforms that realizes the global optimal utility $\max_{A'} R_{A'}(N)$---regardless of whether users are fully non-myopic or have discounted utility. These results have interesting implications for shared data access as an intervention in digital marketplace regulation (e.g. see \cite{cremer2019competition}). At first glance, it would appear that data sharing would resolve the alignment issues, since it prevents platforms from gaining market dominance through data accumulation. However, our results illustrate that the platforms may still not align their actions with user utility at equilibrium. 

\paragraph{Comparison of separate and shared data settings.} To further investigate the efficacy of shared data access as a policy intervention, we compare alignment when the platforms share a data repository to alignment when the platforms have separate data repositories, highlighting two fundamental differences. We focus on the undiscounted setup (Setup \ref{setup:continuoustimeundiscounted}) analyzed in Theorem \ref{thm:BHundiscounted}; in this case, the algorithm class $\AAllContinuous$ satisfies information monotonicity and utility richness (see Lemma \ref{lemma:undiscountedIM}) so the results in Section \ref{sec:alignmentseparate} are also applicable.\footnote{
In the discounted setting, not all of the algorithms in $\AAll$ necessarily satisfy the information monotonicity requirements used in the alignment results for the separate data setting. Thus, Theorem \ref{thm:BH} cannot be used to directly compare the two settings.} The first difference in the nature of alignment is that there is a unique symmetric equilibrium for the shared data setting, which stands in contrast to the range of equilibria that arose in the separate data setting. Thus, while the particularities of equilibrium selection significantly impact alignment in the separate data setting (see Section \ref{sec:discussion}), these particularities are irrelevant from the perspective of alignment in the shared data setting. 

The second difference is that the user quality level of the symmetric equilibrium in the shared data setting is in the \textit{interior} of the range $[\max_{A \in \mathcal{A}} R_A(1), \max_{A \in \mathcal{A}} R_A(N)]$ of user quality levels exhibited in the separate data setting. The alignment in the shared data setting is thus \textit{strictly better} than the alignment of the worst possible equilibrium in the separate data setting. Thus, if we take a pessimistic view of the separate data setting, assuming that the marketplace exhibits the worst-possible equilibrium, then data sharing does help users. On the other hand, the alignment in the shared data setting is also \textit{strictly worse} than the alignment of the best possible equilibrium in the separate data setting. This means if that we instead take an optimistic view of the separate data setting, and assume that the marketplace exhibits this best-case equilibrium, then data sharing is actually harmful for alignment. In other words, when comparing data sharing and equilibrium selection as regulatory interventions, data sharing is worse for users than maintaining separate data and applying an equilibrium selection mechanism that shifts the market towards the best equilibria.

\paragraph{Mechanism for misalignment.} 
Perhaps counterintuitively, the mechanism for misalignment in the shared data setting is that a platform must perfectly align its choice of algorithm with the preferences of a user (given the choices of other users). In particular, the algorithm that is optimal for one user given the actions of other users is different from the algorithm that would be optimal if the users were to cooperate. This is because exploration is costly to users, so users don't want to perform their fair share of exploration, and would rather \textit{free-ride} off of the exploration of other users. As a result, a platform who chooses an algorithm with the global optimal strategy cannot maintain its user base. We formalize this phenomena by establishing a connection with \textit{strategic experimentation}, drawing upon the results of \cite{BH98, BH00, BHSimple} (see Appendix \ref{appendix:strategicexperimentation} for a recap of the relevant results).

\paragraph{Proof sketches of Theorem \ref{thm:BHundiscounted} and Theorem \ref{thm:BH}.} The key insight is that the symmetric equilibria of our game are closely related to the equilibria of the following game $G$. Let $G$ be an $N$ player game where each player chooses an algorithm in $\mathcal{A}$ within the same bandit problem setup as in our game. The players share an information state $\IS$ corresponding to the posterior distributions of the arms. At each time step, all of the $N$ players arrive at the platform, player $i$ pulls the arm drawn from $A_i(\IS)$, and the players all update $\IS$. The utility received by a player is given by their discounted cumulative reward. 

We characterize the symmetric equilibria of the original game for the platforms. 
\begin{restatable}{lemma}{utilityshared}
\label{lemma:utilitysharedinfo}
The solution $(A, A)$ is in equilibrium if and only if $A $ is a symmetric pure strategy equilibrium of the game $G$ described above. 
\end{restatable}
Moreover, the user quality level $Q(A, A)$ is equal to $R_A(N)$, which is also equal to the utility achieved by players in $G$ when they all choose action $A$. 

In the game $G$, the global optimal algorithm $A^* = \argmax_{A' \in \mathcal{A}} R_{A'}(N)$ corresponds to the solution when all $N$ players cooperate rather than arriving at an equilibrium. Intuitively, all of the players choosing $A^*$ is not an equilibrium because exploration comes at a cost to utility, and thus players wish to ``free-ride'' off of the exploration of other players. The value $\max_{A' \in \mathcal{A}} R_{A'}(N)$ corresponds to the cooperative maximal utility that can be obtained the $N$ players.

To show Theorem \ref{thm:BH}, it suffices to analyze structure of the equilibria of $G$. Interestingly, \citet{BH98, BH00, BHSimple}---in the context of strategic experimentation---studied a game very similar to $G$ instantiated in the risky-safe arm bandit problem with algorithm class $\AAll$. We provide a recap of the relevant aspects of their results and analysis in Appendix \ref{appendix:strategicexperimentation}. At a high level, they showed that there is a unique symmetric pure strategy equilibrium and showed that the utility of this equilibrium is strictly below the global optimal. We can adopt this analysis to conclude that the equilibrium player utility in $G$ is strictly below $R_A(N)$. The full proof is deferred to Appendix \ref{sec:proofsshared}.

\subsection{Weak alignment}\label{sec:alignmentshared} 

Although not all values in $[\max_{A'} R_{A'}(1), \max_{A'} R_{A'}(N)]$ can be realized, we show that the user quality level at any symmetric equilibrium is always at least $\max_{A'} R_{A'}(1)$. 
\begin{theorem}
\label{thm:utilitysharedinfo}
Suppose that every algorithm in $\mathcal{A}$ is side information monotonic (Assumption \ref{assumption:IM}). In the shared data setting, at any symmetric equilibrium $(A, A)$, the user quality level $Q(A,A)$ is in the interval $[\max_{A' \in \mathcal{A}} R_{A'}(1), \max_{A' \in \mathcal{A}} R_{A'}(N)] $.
\end{theorem}

Theorem \ref{thm:utilitysharedinfo} demonstrates that the free-riding effect described in Section \ref{sec:counterexample} cannot drive the user quality level below the single-user optimal. Recall that the single-user optimal is also a lower bound on the user quality level for the separate data setting (see Theorem \ref{thm:utilitysepinfo}). This means that regardless of the assumptions on data sharing, the market outcome exhibits a weak form of alignment where the user quality level is at least the single-user optimal.   

\paragraph{Proof sketch of Theorem \ref{thm:utilitysharedinfo}.} We again leverage the connection to the game $G$ described in the proof sketch of Theorem \ref{thm:BH}. The main technical step is to show that at any symmetric pure strategy equilibrium $A$, the player utility $R_A(N)$ is at least $\max_{A' \in \mathcal{A}} R_{A'}(1)$ (Lemma \ref{lemma:utilitygame}). Intuitively, since $A$ is a best response for each player, they must receive no more utility by choosing $A^* \in \argmax_{A' \in \mathcal{A}} R_{A'}(1)$. The utility that they would receive from playing $A^*$ if there were no other players in the game is $R_{A^*}(1) = \max_{A' \in \mathcal{A}} R_{A'}(1)$. The presence of other players can be viewed as background updates to the information state, and the information monotonicity assumption on $A$ guarantees that these updates can only improve the player's utility in expectation. The full proof is deferred to Appendix \ref{sec:proofsshared}.

\section{Algorithm classes $\mathcal{A}$ that satisfy our assumptions}\label{sec:assumptions}
We describe several different bandit setups under which the assumptions on $\mathcal{A}$ described in Section \ref{subsec:assumptions} are satisfied.
\paragraph{Discussion of information monotonicity (Assumption \ref{assumption:IM}).} We first show that in the undiscounted, continuous-time, risky-safe arm bandit setup, the information monotonicity assumptions are satisfied for essentially any algorithm (proof is deferred to Appendix \ref{appendix:assumptions}).  
\begin{lemma}
\label{lemma:undiscountedIM}
Consider the undiscounted, continuous-time risky-safe arm bandit setup (see Setup \ref{setup:continuoustimeundiscounted}). Any algorithm $A \in \AAllContinuous$ satisfies strict information monotonicity and side information monotonicity. 
\end{lemma}

While the above result focuses on undiscounted utility, we also show that information monotonicity can also be achieved with discounting. In particular, information monotonicity is satisfied by  \texttt{ThompsonSampling} (proof is deferred to Appendix \ref{appendix:assumptions}). 
\begin{lemma}
\label{lemma:TS}
For the discrete-time risky-safe arm bandit problem with finite time horizon, prior $p \in (0,1)$, $N = 2$ users, and no background information (see Setup \ref{setup:discreteriskysafe}), \texttt{ThompsonSampling} is strictly information monotonic and side information monotonic for any discount factor $\beta \in (0,1]$. 
\end{lemma}
\noindent In fact, we actually show in the proof of Lemma \ref{lemma:TS} that the $\epsilon$-\texttt{ThompsonSampling} algorithm that explores uniformly with probability $\epsilon$ and applies \texttt{ThompsonSampling} with probability $1-\epsilon$ also satisfies strict information monotonicity and side information monotonicity. 

These information monotonicity assumptions become completely unrestrictive for fully myopic users, where user utility is fully determined by the algorithm's performance at the first time step, before any information updates are made. In particular, \textit{any} algorithm is information constant and side-information monotonic. 

 More broadly, understanding information monotonicity and its variants is crucial
for investigating the incentive properties of bandit algorithms: indeed prior work (e.g. \citet{AMSW20, MSSW16, SS21}) has explored variants of this assumption. Since these works focus on fully myopic users that may arrive at any time step, they require a different information monotonicity assumption, that they call \textit{Bayes monotonicity} \citep{AMSW20}. (An algorithm satisfies Bayes monotonicity if its expected reward is non-decreasing \textit{in time}.) Bayes monotonicity is strictly speaking incomparable to our information monotonicity assumptions; in particular, Bayes monotonicity does not imply either strict information monotonicity or side information monotonicity.

\paragraph{Discussion of utility richness (Assumption \ref{assumption:richness}).} 
At an intuitive level, as long as the algorithm class reflects a range of exploration levels, it will satisfy utility richness. 

We first show that in the undiscounted setup in Theorem \ref{thm:BHundiscounted}, the algorithm class satisfies utility richness (proof in Appendix \ref{appendix:assumptions}). 
\begin{lemma}
\label{lemma:undiscountedUR}
Consider the undiscounted, continuous-time risky-safe arm bandit setup (see Setup \ref{setup:continuoustimeundiscounted}). The algorithm class $\AAllContinuous$ satisfies utility richness. 
\end{lemma}

Since the above result focuses on a particular bandit setup, 
we also describe a general operation to transform an algorithm class into one that satisfies utility richness. In particular, the closure of an algorithm class under mixtures with uniformly random exploration satisfies utility richness (proof in Appendix \ref{appendix:assumptions}).  
\begin{lemma}
\label{lemma:UR}
Consider any discrete-time setup with finite time horizon and bounded mean rewards. For $A \in \mathcal{A}$, let $A_{\epsilon}$ be the algorithm that chooses an arm at random w/ probability $\epsilon$. Suppose that the reward $R_A(N)$ of every algorithm $A \in \mathcal{A}$ is at least $R_{A_1}(N)$ (the reward of uniform exploration), and suppose that the supremum of $\left\{R_A(N) \mid A \in \mathcal{A}\right\}$ is achieved. Then, the algorithm class 
$\mathcal{A}_{\text{closure}} := \left\{A_{\epsilon} \mid A \in \mathcal{A}, \epsilon \in [0,1]\right\}$ satisfies utility richness. 
\end{lemma}

\paragraph{Example classes that achieve information monotonicity and utility richness.} Together, the results above provide two natural bandit setups that satisfy strict information monotonicity, side information monotonicity, and utility richness.
\begin{enumerate}
    \item The algorithm class $\AAllContinuous$  in the undiscounted, continuous-time risky-safe arm bandit setup with any $N \ge 1$ users (see Setup \ref{setup:continuoustimeundiscounted}). 
    \item The class of $\epsilon$-Thompson sampling algorithms in the discrete time risky-safe arm bandit setup with discount factor $\beta \in (0, 1]$, $N = 2$ users, and no background information (see Setup \ref{setup:discreteriskysafe}). 
\end{enumerate}  
These setups, which span the full range of discount factors, provide concrete examples where our alignment results are guaranteed to apply.

\section{Discussion}

Towards investigating competition in digital marketplaces, we present a framework for analyzing competition between two platforms performing multi-armed bandit learning through interactions with a population of users. We propose and analyze the \textit{user quality level} as a measure of the alignment of market equilibria. We show that unlike in typical markets of products, competition in this setting does not perfectly align market outcomes with user utilities, both when the platforms maintain separate data repositories and when the platforms maintain a shared data repository.

Our framework further allows to compare the separate and shared data settings, and we show that the nature of misalignment fundamentally depends on the data sharing assumptions. First, different mechanisms drive misalignment: when platforms have separate data repositories, the suboptimality of an algorithm can be compensated for with a larger user base; when the platforms share data, a platform can't retain its user base if it chooses the global optimal algorithm since users wish to free-ride off of the exploration of other users. Another aspect that depends on the data sharing assumptions is the specific form of misalignment exhibited by market outcomes. The set of realizable user quality levels ranges from the single-user optimal to the global optimal in the separate data setting; on the other hand, in the shared data setting, neither of these endpoints may be realizable. These differences suggest that data sharing performs worse as a regulatory intervention than a well-designed equilibrium selection mechanism. 

More broadly, our work provides a mathematical explanation  of phenomena documented in recent policy reports and reveals that competition has subtle consequences for users in digital marketplaces that merit further inquiry. We hope that our work provides a starting point for building a theoretical foundation for investigating competition and designing regulatory interventions in digital marketplaces. 

\section{Acknowledgments}

We would like to thank Anca Dragan, Yannai Gonczarowski, Erik Jones, Rad Niazadeh, Jacob Steinhardt, Jonathan Stray, Nilesh Tripuraneni, Abhishek Shetty, and Alex Wei for helpful comments on the paper. 
This work is in part supported by National Science Foundation under grant CCF-2145898, the Mathematical Data Science program of the Office of Naval Research under grant number N00014-18-1-2764, the Vannevar Bush Faculty Fellowship program under grant number N00014-21-1-2941, a C3.AI Digital Transformation Institute grant, the Paul and Daisy Soros Fellowship,  and the Open Phil AI Fellowship. 

\bibliographystyle{plainnat}
\bibliography{bibliography.bib}

\newpage

\appendix

\section{Example bandit setups}

We consider the following risky-safe arm setups in our results. The first setup is a risky-safe arm bandit setup in continuous time, where user rewards are undiscounted. 
\begin{setup}[Undiscounted, continuous time risky-safe arm setup]
\label{setup:continuoustimeundiscounted}
Consider a risky-safe arm bandit setup where the algorithm class is  
\[\AAllContinuous := \left\{A_f \mid f: [0,1] \rightarrow [0,1] \text{ is measurable}, f(0) = 0, f(1) = 1, f \text { is continuous  at $0$ and $1$ } \right\}.\] The bandit setup is in continuous time: if a platform chooses algorithm $A \in \AAllContinuous$, then at a given time step with information state $\mathcal{I}$, the user of that platform devotes a $\mathbb{P}[A(\mathcal{I}) = 1]$ fraction of the time step to the risky arm and the remainder of the time step to the safe arm. Let the prior be initialized so $p_0 := p(\mathcal{I}) = \mathbb{P}_{X \sim \DPrior_1}[X = h] \in (0,1)$. Let the rewards be such that the full-information payoff $hp_0 + s(1-p_0) = 0$. Let the background information quality be $\sigma_b < \infty$. Let the time horizon $T = \infty$ be infinite, and suppose the user utility is undiscounted.\footnote{Formally, this means that the user utility is the limit $\lim_{T \rightarrow \infty} T \cdot \mathbb{E}\left[\frac{1}{T} \int d\pi(t)\right]$ as the time horizon goes to $\infty$, or alternatively the limit $\lim_{\beta \rightarrow 0} T \cdot \mathbb{E}\left[\int e^{-\beta t} d\pi(t)\right]$ as the discount factor vanishes. See \citet{BH00} for a justification that these limits are well-defined.} 
\end{setup}
 The next setup is again a risky-safe arm bandit setup in continuous time, but this time with discounted rewards. 
\begin{setup}[Discounted, continuous time risky-safe arm setup]
\label{setup:continuoustimediscounted}
Consider a risky-safe arm bandit setup where the algorithm class is  $\AAllContinuous$. The bandit setup is in continuous time: if a platform chooses algorithm $A \in \AAllContinuous$, then at a given time step with information state $\mathcal{I}$, the user of that platform devotes a $\mathbb{P}[A(\mathcal{I}) = 1]$ fraction of the time step to the risky arm and the remainder of the time step to the safe arm. Let the high reward $h$ be $1$, the low reward $l$ be $0$, and let the prior be initialized to some $p(\mathcal{I}) \ge \mathbb{P}_{X \sim \DPrior_1}[X = 1] > s$ where $s$ is the safe arm reward. Let the time horizon $T = \infty$ be infinite, suppose that there is no background information $\sigma_b = \infty$, and suppose the user utility is discounted with discount factor $\beta \in (0, \infty)$. 
\end{setup}
Finally, we consider another discounted risky-safe bandit setup, but this time with discrete time and finite time horizon. 
\begin{setup}[Discrete, risky-safe arm setup]
\label{setup:discreteriskysafe}
Consider a risky-safe arm bandit setup where the algorithm class is $\mathcal{A} = \left\{A_{f^{TS}_{\epsilon}} \mid \epsilon \in [0,1]\right\}$, where $A_{f^{TS}_{\epsilon}}$ denotes the $\epsilon$-Thompson sampling algorithm given by $f^TS_{\epsilon}(p) = \epsilon + (1-\epsilon)p$. The bandit setup is in discrete time: if a platform chooses algorithm $A \in \mathcal{A}$, then at a given time step with information state $\mathcal{I}$, the user of that platform chooses the risky arm with probability $\mathbb{P}[A(\mathcal{I}) = 1]$ and the safe arm with probability  $\mathbb{P}[A(\mathcal{I}) = 0]$. Let the time horizon $T < \infty$ be finite, suppose that the user utility is discounted with discount factor $\beta \in (0, 1]$, that there is no background information $\sigma_ b =0$, and let the prior be initialized to $p(\mathcal{I}) \in (0,1)$
\end{setup}

\section{Further details about the model choice}
We examine two aspects our model---the choice of equilibrium set $\mathcal{E}$ and the action space of users---in greater detail. 

\subsection{What would change if users can play mixed strategies?}\label{sec:mixed}
Suppose that $\mathcal{E}_{A_1, A_2}$ were defined to be the set of \textit{all} equilibria for the users, rather than only pure strategy equilibria. The main difference is that all users might no longer choose the same platform at equilibrium, which would change the nature of the set $\mathcal{E}_{A_1, A_2}$. In particular, even when both platforms choose the same algorithm $A$, there is a symmetric mixed equilibrium where all users randomize equally between the two platforms. At this mixed equilibrium, the utility of the users is $\mathbb{E}_{X \sim Bin(N, 1/2)} [R_A(X)]$, since the number of users at each platform would follows a binomial distribution. This quantity might be substantially lower than $R_A(N)$   depending on the nature of the bandit algorithms. As a result, the user quality level $Q(A, A)$, which is measured by the \textit{worst} equilibrium for the users in $\mathcal{E}$, could be substantially lower than $R_A(N)$. Moreover, the condition for $(A, A)$ to be an equilibrium for the platforms would still be that $R_A(N) \ge \max_{A'} R_{A'}(1)$, so there could exist a platform equilibria with user quality level much lower than $\max_{A'} R_{A'}(1)$. Intuitively, the introduction of mixtures corresponds to users no longer coordinating between their choices of platforms----this leads to no single platform accumulating all of the data, thus lowering user utility.

\subsection{What would change if users could change platforms at each round?}\label{sec:users}

Our model assumes that users choose a platform at the beginning of the game which they participate on for the duration of the game. In this section, we examine this assumption in greater detail, informally exploring what would change if the users could switch platforms.

First, we provide intuition that in the shared data setting, there would be no change in the structure of the equilibrium as long as the equilibrium class $\mathcal{A}$ is closed under mixtures (i.e. if $A_1, A_2 \in \mathcal{A}$, then the algorithm that plays $A_1$ with probability $p_1$ and $A_2$ with probability $p_2$ must be in $\mathcal{A}$). A natural model for users switching platforms would be that users see the public information state at every round and choose a platform based on this information state (and algorithms for the platforms). A user's strategy is thus a mapping from an information state $\mathcal{I}$ to $\left\{1,2\right\}$, and the platform would receive utility for a user depending on the fraction of time that they spend on that platform. Suppose that symmetric (mixed) equilibria for users are guaranteed to exist for any choice of platform algorithms, and we define the  platform's utility by the minimal number of (fractional) users that they receive at any symmetric mixed equilibrium. In this model, we again see that $(A,A)$ is a symmetric equilibrium for the platform if and only if $A$ is an symmetric pure strategy equilibrium in the game $G$ defined in Section \ref{sec:separate}. (To see this, note if $A$ is not a symmetric pure strategy equilibrium, then the platform can achieve higher utility by choosing $A'$ that is a deviation for a player in the game $G$. If $(A,A)$ is a symmetric pure strategy equilibrium, then ). Thus, the alignment results will remain the same.

In the separate data setting, even defining a model where users can switch platforms is more subtle since it is unclear how the information state of the users should be defined. One possibility would be that each user keeps track of their own information state based on the rewards that they observe. Studying the resulting equilibria would require reasoning about the evolution of user information states and furthermore may not capture practical settings where users see the information of other users. Given these challenges, we defer the analysis of users switching platforms in the case of separate data to future work.

\section{Proof of Theorem \ref{thm:idealized}}\label{appendix:proofidealized}

We prove Theorem \ref{thm:idealized}.
\begin{proof}[Proof of Theorem \ref{thm:idealized}]
We split into two cases: (1) either $R_{A_1}(1) = \max_{A'} R_{A'}(1)$ or $R_{A_2}(1) = \max_{A'} R_{A'}(1)$, and (2) $R_{A_1}(1) < \max_{A'} R_{A'}(1)$ or $R_{A_2}(1) < \max_{A'} R_{A'}(1)$. 

\paragraph{Case 1: $R_{A_1}(1) = \max_{A'} R_{A'}(1)$ or $R_{A_2}(1) = \max_{A'} R_{A'}(1)$.} We show that $(A_1, A_2)$ is an equilibrium. 

Suppose first that $R_{A_1}(1) = \max_{A'} R_{A'}(1)$ and $R_{A_2}(1) = \max_{A'} R_{A'}(1)$. We see that the strategies $\vectorfont{p} = [1]$ and $\vectorfont{p} = [2]$, where the user chooses platform 1, is in the set of equilibria $\mathcal{E}_{A_1, A_2}$. This means that $v_1(A_1; A_2) = v_1(A_2; A_1) = 0$. Suppose that platform 1 chooses another algorithm $A''$. Since $R_{A''}(1) \le \max_{A'} R_{A'}(1)$, we see that $\vectorfont{p} = [2, \ldots, 2]$ is still an equilibrium. Thus, $v_1(A'; A_2) = 0$. This implies that $A_1$ is a best response for platform 1, and an analogous argument shows $A_2$ is a best response for platform 2. When the platforms choose $(A_1, A_2)$, at either of the user equilibria $\vectorfont{p} = [1]$ or $\vectorfont{p} = [2]$, the user utility is $\max_{A'} R_{A'}(1)$. Thus $Q(A_1, A_2) = \max_{A'} R_{A'}(1)$. 

Now, suppose that exactly one of $R_{A_1}(1) = \max_{A'} R_{A'}(1)$ and $R_{A_2}(1) = \max_{A'} R_{A'}(1)$ holds. 
WLOG, suppose $R_{A_1}(1) = \max_{A'} R_{A'}(1)$. Since $R_{A_2}(1) < \max_{A'} R_{A'}(1)$, we see that $[2] \not\in \mathcal{E}_{A_1, A_2}$. On the other hand, $[1] \in \mathcal{E}_{A_1, A_2}$. This means that $v_1(A_1; A_2) = 1$ and $v_2(A_2; A_1) = 0$. Thus, $A_1$ is a best response for platform 1 trivially because $v_1(A'; A_2) \le 1$ for all $A'$ by definition. We next show that $A_2$ is a best response for platform 2. If the platform 2 plays another algorithm $A'$, then $[1]$ will still be in equilibrium for the users since platform 1 offers the maximum possible utility. Thus, $v_1(A'; A) = 0$, and $A_2$ is a best response for platform 2. When the platforms choose $(A_1, A_2)$, the only user equilibria is $\vectorfont{p} = [1]$ where the user utility is $\max_{A'} R_{A'}(1)$. Thus $Q(A_1, A_2) = \max_{A'} R_{A'}(1)$. 

\paragraph{Case 2: $R_{A_1}(1) < \max_{A'} R_{A'}(1)$ or $R_{A_2}(1) < \max_{A'} R_{A'}(1)$.} It suffices to show that $(A_1, A_2)$ is not an equilibrium. WLOG, suppose that
$R_{A_2}(1) \le R_{A_1}(1)$. We see that $[1] \in \mathcal{E}_{A_1, A_2}$. Thus, $v_2(A_2; A_1) = 0$. However, if platform 2 switches to $A'' \in \argmax_{A' \in \mathcal{A}} R_{A'}(1)$, then $\mathcal{E}_{A_1, A''}$ is equal to $\left\{[2]\right\}$ and so $v_2(A'';A_1) = 1$. This means that $A_2$ is not a best response for platform 2, and thus $(A_1, A_2)$ is not an equilibrium.
\end{proof}

\section{Proofs for Section \ref{sec:separate}}\label{sec:proofsseparate}

In the proofs of Theorems \ref{thm:realizable} and \ref{thm:utilitysepinfo}, the key technical ingredient is that pure strategy equilibria for users take a simple form. In particular, under strict information monotonicity, we show that in every pure strategy equilibrium $p^* \in \mathcal{E}^{A_1, A_2}$, all of the users choose the same platform. 
\begin{lemma}
\label{lemma:pure}
Suppose that every algorithm $A \in \mathcal{A}$ is either strictly information monotonic or information constant (see Assumption \ref{assumption:IM}). For any choice of platform algorithms $A_1, A_2 \in \mathcal{A}$ such that at least one of $A_1$ and $A_2$ is strictly information monotonic, it holds that:
\[\mathcal{E}_{A_1, A_2} \subseteq \left\{[1, \ldots, 1], [2, \ldots, 2]\right\}.\]
\end{lemma}
\begin{proof}
WLOG, assume that $A_1$ is strictly information monotonic. Assume for sake of contradiction that the user strategy profile $[1, \ldots, 1, 2, \ldots, 2]$ (with $N_1 > 0$ users choosing platform 1 and $N_2 > 0$ users choosing platform 2) is in $\mathcal{E}_{A_1, A_2}$. Since $[1, \ldots, 1, 2, \ldots, 2]$ is an equilibrium, a user choosing platform 1 not want to switch to platform 2. The utility that they currently receive is $R_{A_1}(N_1)$ and the utility that they would receive from switching is  $R_{A_2}(N_2 + 1)$, so this means:
\[R_{A_2}(N_2 + 1) \le R_{A_1}(N_1). \] 
Similarly, since $[1, \ldots, 1, 2, \ldots, 2]$ is an equilibrium, a user choosing platform 2 not want to switch to platform 1. The utility that they currently receive is $R_{A_2}(N_2)$ and the utility that they would receive from switching is  $R_{A_1}(N_1 + 1)$, so this means:
\[R_{A_1}(N_1 + 1) \le R_{A_2}(N_2). \] 
Putting this all together, we see that:
\[R_{A_2}(N_2 + 1) \le R_{A_1}(N_1) < R_{A_1}(N_1 + 1) \le R_{A_2}(N_2), \]
which is a contradiction since $A_2$ is either strictly information monotonic or information constant. 
\end{proof}

\subsection{Proof of Theorem \ref{thm:realizable}}
We prove Theorem \ref{thm:realizable}.
\begin{proof}[Proof of Theorem \ref{thm:realizable}]
Since the algorithm class $\mathcal{A}$ is utility rich (Assumption \ref{assumption:richness}), we know that for any $\alpha \in [\max_{A' \in \mathcal{A}} R_{A'}(1), \max_{A' \in \mathcal{A}} R_{A'}(N)]$, there exists an algorithm $A^* \in \mathcal{A}$ such that $R_{A^*}(N) = \alpha$. We claim that $(A^*, A^*)$ is an equilibrium and we show that $Q(A^*, A^*) = \alpha$. 

To show that $(A^*, A^*)$ is an equilibrium, suppose that platform 1 chooses any algorithm $A \in \mathcal{A}$. We claim that $[2, 2, \ldots, 2] \in \mathcal{E}_{A_1, A_2}$. To see this, notice that the utility that a user receives from choosing platform 2 is $R_{A^*}(N)$, and the utility that they would receive if they deviate to platform is $R_{A^*}(1)$. By definition, we see that:
\[R_{A^*}(N) = \alpha \ge \max_{A' \in \mathcal{A}} R_{A'}(1) \ge R_{A^*}(1) ,\]
so choosing platform 2 is a best response for the user. 
This means that $v_1(A; A_2) = 0$ for any algorithm $A \in \mathcal{A}$. This in particular means that $A^*$ is a best response for platform 1. By an analogous argument, we see that $A^*$ is a best response for platform 2, and so $(A^*, A^*)$ is an equilibrium. 

We next show that the user quality level is $\alpha$. It suffices to  examine the set of pure strategy equilibria for users when the platforms play $(A^*, A^*)$. By assumption, either $A^*$ is information constant or $A^*$ is strictly information monotonic. If $A^*$ is information constant, then $R_{A^*}(n)$ is constant in $n$. This means that any pure strategy equilibrium $\vectorfont{p}$ generates utility $\alpha$ for all users, so $Q(A^*, A^*) = \alpha$ as desired. If $A^*$ is strictly information monotonic, then we apply Lemma \ref{lemma:pure} to see that $\mathcal{E}_{A^*, A^*} \subseteq \left\{[1, \ldots, 1], [2, \ldots, 2]\right\}$. In fact, since the platforms play the same algorithm, we see that 
\[\mathcal{E}_{A^*, A^*} = \left\{[1, \ldots, 1], [2, \ldots, 2]\right\}.\]
The user utility at these equilibria is $R_{A^*}(N) = \alpha$, so $Q(A^*, A^*) = \alpha$ as desired. 
\end{proof}

\subsection{Proof of Theorem \ref{thm:utilitysepinfo}}

We prove Theorem \ref{thm:utilitysepinfo}. 
\begin{proof}[Proof of Theorem \ref{thm:utilitysepinfo}]
First, we show the upper bound of $Q(A_1, A_2) \le \max_{A' \in \mathcal{A}} R_{A'}(N)$. In fact, we show this upper bound for any selection of user actions $\vectorfont{p}$, user $1 \le i \le N$, and platform algorithms $(A_1, A_2)$. If a user chooses platform 1 and $0 \le n \le N - 1$ other users also choose platform 1, then the user's utility is $R_{A_1}(n+1)$ which by Assumption \ref{assumption:IM} can be upper bounded by $R_{A_1}(N) \le \max_{A' \in \mathcal{A}} R_{A'}(N)$. Similarly, if a user chooses platform 2, their utility can also be upper bounded by $R_{A_2}(N-n) \le R_{A_2}(N) \le \max_{A' \in \mathcal{A}} R_{A'}(N)$. This establishes the desired upper bound. 

The remainder of the proof boils down to lower bounding $Q_{A_1, A_2}$ at any equilibrium $(A_1, A_2)$ by $\max_{A' \in \mathcal{A}} R_{A'}(1)$. We divide into two cases: (1) $R_{A_1}(N) < \max_{A' \in \mathcal{A}} R_{A'}(1)$ and $R_{A_2}(N)) < \max_{A' \in \mathcal{A}} R_{A'}(1)$, (2) at least one of $R_{A_1}(N) \ge \max_{A' \in \mathcal{A}} R_{A'}(1)$ and $R_{A_2}(N) \ge \max_{A' \in \mathcal{A}} R_{A'}(1)$ holds. 

\paragraph{Case 1: $R_{A_1}(N) < \max_{A' \in \mathcal{A}} R_{A'}(1)$ and $R_{A_2}(N)) < \max_{A' \in \mathcal{A}} R_{A'}(1)$.} We show that $(A_1,A_2)$ is not an equilibrium. WLOG suppose that $R_{A_2}(N) \le R_{A_1}(N)$. 
First, we claim that $v_2(A_2; A_1) = 0$. It suffices to show that $[1, \ldots, 1] \in \mathcal{E}_{A_1, A_2}$. To see this, notice that the utility that a user derives from choosing platform 1 is $R_{A_1}(N)$, while the utility that a user would derive from choosing platform 2 is $R_{A_2}(1)$. Moreover, we have that: 
\[R_{A_1}(N) \ge R_{A_2}(N) \ge R_{A_2}(1), \]
since $A_2$ is either strictly information monotonic or information constant by assumption. This means that choosing platform 1 is a best response for the user, so $[1, \ldots, 1] \in \mathcal{E}_{A_1, A_2}$.  

Next, we claim that $A_2$ is not a best response for platform 2. It suffices to show that for $A'' \in \argmax_{A' \in \mathcal{A}} R_{A'}(1)$, platform 2 receives utility $v_2(A''; A_1) > v_2(A_2; A_1) = 0$. It suffices to show that $[1, \ldots, 1] \not\in \mathcal{E}_{A_1, A''}$. To see this, notice that the utility that a user derives from choosing platform 1 is $R_{A_1}(N)$, while the utility that a user would derive from choosing platform 2 is $R_{A''}(1)$. Moreover, we have that: 
\[R_{A_1}(N) < \max_{A' \in \mathcal{A}} R_{A'}(1) = R_{A''}(1), \]
so choosing platform 1 is not a best response for the user. Thus, $[1, \ldots, 1] \not\in \mathcal{E}_{A_1, A''}$ and 
$v_2(A''; A_1) > 0$. 

This means that $(A_1, A_2)$ is not an equilibrium as desired. 

\paragraph{Case 2: at least one of $R_{A_1}(N) \ge \max_{A' \in \mathcal{A}} R_{A'}(1)$ and $R_{A_2}(N) \ge \max_{A' \in \mathcal{A}} R_{A'}(1)$ holds.} WLOG suppose that $R_{A_1}(N) \ge \max_{A' \in \mathcal{A}} R_{A'}(1)$. We break into 2 cases: (1) at least one of $A_1$ and $A_2$ is strictly information monotonic, and (2) both $A_1$ and $A_2$ are information constant.

\paragraph{Subcase 1: at least one of $A_1$ and $A_2$ is strictly information monotonic.} It suffices to show that at every equilibrium in $\mathcal{E}_{A_1, A_2}$, the user utility is at least $\max_{A' \in \mathcal{A}} R_{A'}(1)$. We can apply Lemma \ref{lemma:pure}, which tells us that $\mathcal{E}_{A_1, A_2} \subseteq \left\{[1, \ldots, 1],[2, \ldots, 2] \right\}$. 

At $[1, \ldots, 1]$, the user utility is $R_{A_1}(N) \ge \max_{A' \in \mathcal{A}} R_{A'}(1)$ as desired.  

Suppose now that $[2, \ldots, 2] \in \mathcal{E}_{A_1, A_2}$. The user utility at this equilibria is $R_{A_2}(N)$, so it suffices to show that $R_{A_2}(N) \ge \max_{A' \in \mathcal{A}} R_{A'}(1)$. Since $[2, \ldots, 2] \in \mathcal{E}_{A_1, A_2}$, we see that $v_1(A_1; A_2) = 0$. Thus, if platform 1 changed to $A \in \argmax_{A' \in \mathcal{A}} R_{A'}(1)$, it must hold that $v_1(A; A_2) \le v_1(A_1; A_2) = 0$ since $(A_1, A_2)$ is a platform equilibrium. This means that $[2, \ldots, 2] \in \mathcal{E}_{A, A_2}$. The utility that a user would derive from choosing platform 1 is $R_{A}(1)$, while the utility that a user derives from choosing platform 2 is $R_{A_2}(N)$. Since $[2, \ldots, 2]$ is an equilibrium, it must hold that 
\[R_{A_2}(N) \ge R_A(1) = \max_{A' \in \mathcal{A}} R_{A'}(1).\] 
as desired. 

\paragraph{Subcase 2: both $A_1$ and $A_2$ are information constant.} Regardless of the actions of other users, if a user chooses platform 1 they receive utility $R_{A_1}(N)$ and if a user chooses platform 2 they receive utility $R_{A_2}(N)$. Thus, if any user chooses platform 1 in equilibrium, then $R_{A_1}(N) \ge R_{A_2}(N)$; and if any user chooses platform 2 in equilibrium, then $R_{A_2}(N) \ge R_{A_1}(N)$. The quality level $Q(A_1, A_2)$ is thus $\max(R_{A_1}(N), R_{A_2}(N)) \ge \max_{A' \in \mathcal{A}} R_{A'}(1)$ as desired. 
\end{proof}

\section{Proofs for Section \ref{sec:shared}}\label{sec:proofsshared}

To analyze the equilibria in the shared data setting, we relate the equilibria of our game to the equilibria of an $N$-player game $G$ that is closely related to strategic experimentation \cite{BH98, BH00, BHSimple}. In Appendix \ref{appendix:prooflemma}, we formally establish the relationship between the equilibria in $G$ and the equilibria in our game. In Appendix \ref{appendix:strategicexperimentation}, we provide a recap of the results from strategic experimentation literature for the risky-safe arm bandit problem \citep{BH00, BH98}. In Appendices \ref{appendix:proofBHundiscounted}-\ref{appendix:utilityshared}, we prove our main results using these tools.

\subsection{Proof of Lemma \ref{lemma:utilitysharedinfo}}\label{appendix:prooflemma}

In Lemma \ref{lemma:utilitysharedinfo} (restated below), we show that the symmetric equilibria of our game are equivalent to the symmetric pure strategy equilibria of the following game. Let $G$ be an $N$ player game where each player chooses an algorithm in $\mathcal{A}$ within the same bandit problem setup as in our game. The players share an information state $\IS$ corresponding to the posterior distributions of the arms. At each time step, all of the $N$ users arrive at the platform, player $i$ pulls the arm drawn from $A_i(\IS)$, and the players all update $\IS$. The utility received by a player is given by their expected discounted cumulative reward. 
\utilityshared*
\begin{proof}
We first claim that $\mathcal{E}_{A, A}$ is equal to the set of all possible strategy profiles $\vectorfont{p}$. To see this, notice that the user utility is $R_{A}(N)$ regardless of their action or the actions of other users. Thus, \textit{any} strategy profile is in equilibrium. 

Suppose first that $A$ is a symmetric pure strategy equilibrium of the game $G$. Consider the solution $(A, A)$---since  $\mathcal{E}_{A, A}$ contains $[2, \ldots, 2]$, we see that $v_1(A; A) = 0$. Suppose that platform 1 instead chooses $A' \neq A$. We claim that $[2, \ldots, 2] \in \mathcal{E}_{A', A}$. To see this, when all of the other users choose $A$, then no user wishes to deviate to any other algorithm, including $A'$, since $A$ is a symmetric pure strategy equilibrium in $G$. Thus, $v_1(A'; A) = 0 = v_1(A; A$, and $A$ is a best response for platform 1 as desired. An analogous argument applies to platform 1, thus showing that $(A,A)$ is an equilibrium. 

To prove the converse, suppose that $(A,A)$ is in equilibrium. Since  $\mathcal{E}_{A, A}$ contains $[2, \ldots, 2]$, we see that $v_1(A; A) = 0$. It suffices to show that for every $A' \in \mathcal{A}$, it holds that $v_1(A'; A) = 0 = v_1(A; A)$. Assume for sake of contradiction that $A$ is not a symmetric pure strategy equilibrium in the game $G$. Then there exists a deviation $A'$ for which a user achieves strictly higher utility than the algorithm $A$ in the game $G$. Suppose that platform 1 chooses $A'$. Then we see that $\mathcal{E}_{A', A}$ does \textit{not} contain $[2, \ldots, 2]$, since a user will wish to deviate to platform 1. This means that $v_1(A'; A) > 0$ which is a contradiction. 
\end{proof}

\subsection{Recap of strategic experimentation in the risky-safe arm bandit problem \cite{BH98, BH00, BHSimple}}\label{appendix:strategicexperimentation}

To analyze the game $G$, we leverage the literature on strategic experimentation for the risky-safe arm problem. We provide a high-level summary of results from \cite{BH98, BH00, BHSimple}, deferring the reader to \cite{BH98, BH00, BHSimple} for a full treatment.

\citet{BH98, BH00, BHSimple} study an infinite time-horizon, continuous-time game with $N$ users updating a common information source, focusing on the risky-safe arm bandit problem. The observed rewards $\DNoise$ are given by the mean reward with additive noise with variance $\sigma$. All $N$ users arrive at the platform at every time step and choose a \textit{fraction} of the time step to devote to the risky arm (see \citet{BH98} for a discussion of how this relates to mixed strategies). The players receive background information according to $N(\sqrt{\zeta} h, \sigma^2)$. (By rescaling, we can equivalently think of this as background noise from $N(h, \sigma_b^2)$ where $\sigma_b = \sigma / \sqrt{\zeta}$.) Each player's utility is defined by the (discounted) rewards of the arms that they pull. 

\citet{BH98, BH00, BHSimple} study the Markov equilibria of the resulting game, so user strategies correspond to mappings from the posterior probability that the risky arm has high reward to the fraction of the time step devoted to the risky arm. We denote the user strategies by measurable functions $f:[0,1] \rightarrow [0,1]$. The function $f$ is a symmetric pure strategy equilibrium if for any choice of prior, $f$ is optimal for each user in the game with that choice of prior.

\paragraph{Undiscounted setting.} The cleanest setting is the case of undiscounted rewards. To make the undiscounted user utility over an infinite time-horizon well-defined, the full-information payoff is subtracted from the user utilities. For simplicity, let's focus on the setting where the full information payoff is $0$. In this case, the undiscounted user utility is equal to $\mathbb{E}\left[\int d\pi(t)\right]$, where $d \pi(t)$ denotes the payoff received by the user (see \citet{BHSimple} for a justification that this expectation is well-defined). 

With undiscounted rewards, the user utility achieved by a set of strategies $f_1, \ldots, f_n$ permits a clean closed-form solution. 
\begin{lemma}[Informal restatement of results from \citet{BHSimple}]
\label{lemma:closedform}
Suppose that $N$ players choose the strategies $f_1, \ldots, f_N$ respectively. If the prior is $p$, then the utility $K(p; f_1, \ldots, f_n)$ of player 1 is equal to:
\[
  K(p; f_1, \ldots, f_N) = \int G(p,q) \frac{(1-f_1(q)) s + (q h + (1-q) l) f_1(q) - (hq + s(1-q)) }{\frac{\sigma^2}{\sigma_b^2} + \sum_{i=1}^N f_i(q)} dq   
\]
 where:
 \[
  G(p,q) = 
 \begin{cases}
\frac{2 \sigma^2 p}{(h-l)^2 q^2 (1-q)} &\text{ if  } p \le q \\
\frac{\sigma^2 (1-p)}{(h-l)^2q(1-q)^2} & \text{ if } p \ge q. 
 \end{cases}
 \]
\end{lemma}
\begin{proof}[Proof sketch]
We provide a proof sketch, deferring the full proof to \citet{BH00}. For ease of notation, we let $K(p)$ denote $K(p; f_1, \ldots, f_N)$. The change in the posterior when the posterior is $p$ is mean $0$ and variance $\left(\frac{\sigma^2}{\sigma^2_b} + \sum_{i=1}^n f_i(p)\right) \Phi(p)$ where $\Phi(p) = \left(\frac{p(1-p)(h-l)}{\sigma} \right)^2$. The utility of player 1 is equal to the sum of the current payoff and the continuation payoff.  It can be shown that this is equal to:  
 \[K(p) + \left[(1-f_1(p)) s + (p h + (1-p) l) f_1(p) - (hp + s(1-p)) + \left(\frac{\sigma^2}{\sigma^2_b} +  \sum_{i=1}^n f_i(p) \right) \Phi(p) \frac{K''(p)}{2}\right] dt.\]
 This means that:
 \[0 = (1-f_1(p)) s + (p h + (1-p) l) f_1(p) - (hp + s(1-p))+ \left(\frac{\sigma^2}{\sigma^2_b} +  \sum_{i=1}^n f_i(p) \right) \Phi(p) \frac{K''(p)}{2}.\]
 We can directly solve this differential equation to obtain the desired expression. 
\end{proof}

\citet{BH00, BHSimple} characterize the symmetric equilibria, and  
we summarize below the relevant results for our analysis. 
\begin{theorem}[Informal restatement of results from \cite{BH00, BHSimple}]
\label{thm:restatementundiscounted}
In the undiscounted game, there is a unique symmetric equilibrium $f^*$. When the prior is in $(0,1)$, the equilibrium utility is strictly smaller than the optimal utility in the $N$ user team problem (where users coordinate) and strictly larger than the optimal utility in the $1$ user team problem. 
\end{theorem}
\begin{proof}[Proof sketch]
The existence of equilibrium boils down to characterizing the best response of a player given the experimentation levels of other players. This permits a full characterization of all equilibria (see \citet{BHSimple}) from which we can deduce that there is a unique symmetric equilibrium $f^*$. Moreover, $f^*(p)$ is $0$ if $p$ is sufficiently low, $1$ if $p$ is sufficiently high, and interpolates between $0$ and $1$ for intermediate values of $p$. 

To see that the utility when all users play $f^*$ is strictly smaller than that of the optimal utility in the $N$ user team problem, the calculations in \citet{BH00} show that there is a unique strategy $f_N$ that achieves the $N$-user team optimal utility and this is a cutoff strategy. On the other hand, $f^*$ is not a cutoff strategy, so it cannot achieve the team optimal utility. 

To see that the utility when all users play $f^*$ than that that of the $1$ user team problem, suppose that a user instead chooses the optimal 1 user strategy $f_1$. Using Lemma \ref{lemma:closedform} and the fact that $f^*$ involves nonzero experimentation at some posterior probabilities $p$, we see that the user would achieve strictly higher utility than in a single-user game where they play $f'$. This means that playing $f_1$ results in utility strictly higher than the single-user optimal utility. Since $f^*$ is a best response, this means that playing $f^*$ also results in strictly higher utility than the single-user optimal. 
\end{proof}

\paragraph{Discounted setting.} With discounted rewards, the analysis turns out to be much more challenging since the user utility does not permit a clean closed-form solution. Nonetheless, \citet{BH98} are able to prove properties about the equilibria.We summarize below the relevant results for our analysis.
\begin{theorem}[Informal restatement of results from \cite{BH98}]
\label{thm:restatement}
In the undiscounted game when there is no background information ($\sigma_b = \infty$), the following properties hold: 
\begin{enumerate}
    \item For any choice of discount factor, there is a \textit{unique} symmetric equilibrium $f^*$.
    \item The function $f^*$ is monotonic and strictly increasing from $0$ to $1$ in an interval $[c_{low}, c_{high}]$ where $c_{low} < c_{high} < s$. 
    \item When the prior is initialized above $c_{low}$, the posterior never reaches $c_{low}$ but converges to $c_{low}$ in the limit.
    \item When the prior is initialized above $c_{low}$, the equilibrium utility of $f^*$ is at least as large as the optimal utility in the 1 user problem. The  equilibrium utility of $f^*$ is strictly smaller than the optimal utility in the $N$ user team problem (where users coordinate). 
\end{enumerate}
\end{theorem}
\noindent  We note that we define the discounted utility as $\mathbb{E}\left[\int e^{- \beta t} d\pi(t)\right]$, while \citet{BH98}  defines the discounted utility as $\mathbb{E}\left[\beta \int e^{- \beta t} d\pi(t)\right]$; however, this constant factor difference does not make a difference in any of the properties in Theorem \ref{thm:restatement}.

\subsection{Proof of Theorem \ref{thm:BHundiscounted}}\label{appendix:proofBHundiscounted}

\begin{proof}[Proof of Theorem \ref{thm:BHundiscounted}]

The key ingredient in the proof of Theorem \ref{thm:BHundiscounted} is Lemma \ref{lemma:utilitysharedinfo}. This result shows that $(A, A)$ is an equilibrium for the platforms if and only if $A$ is a symmetric pure strategy equilibrium in the game $G$. Moreover, we see that at any equilibrium in $\mathcal{E}_{A, A}$, the users achieve utility $R_{A}(N)$ since the information state is shared between platforms. Notice that $R_{A}(N)$ is also equal to the utility that users achieve in the game $G$ if they all choose $A$. It thus suffices to show that the equilibrium utility is a unique value $\alpha^* \in (\max_{A' \in \mathcal{A}} R_{A'}(1), \max_{A' \in \mathcal{A}} R_{A'}(N))$ at every symmetric pure strategy equilibrium in $G$. 

To analyze the game $G$, we leverage the results in the literature on strategic experimentation (see Theorem \ref{thm:restatementundiscounted}). We know by Theorem \ref{thm:restatementundiscounted} that there is a unique equilibrium $f^*$ in the undiscounted strategic experimentation game. However, the equilibrium concepts are slightly different because $f$ is a symmetric equilibrium in the undiscounted strategic experimentation game if for \textit{any} choice of prior, $f$ is optimal for each user in the game with that choice of prior; on the other hand, $f$ is a symmetric equilibrium in $G$ if for the specific choice of prior of the bandit setup, $f$ is optimal for each user in the game with that choice of prior. We nonetheless show that $f^*$ is the unique symmetric equilibrium in the game $G$. 

We first claim that the algorithm $A_{f^*}$ is a symmetric pure strategy equilibrium in $G$. To see this, notice that if a user in $G$ were to deviate to $A_{f'} \in \AAllContinuous$ for some $f'$ and achieve higher utility, then it would also be better for a user in the game in \cite{BH98} to deviate to $f'$ when the prior is $\mathbb{P}_{X \sim \DPrior_1}[X = 1]$. This is not possible since $f$ is an equilibrium in $G$, so $A_{f^*}$ is a symmetric pure strategy equilibrium in $G$. This establishes that an equilibrium exists in $G$. 

We next claim that if $A_f$ is a symmetric pure strategy equilibrium in $G$, then $f(c) = f^*(c)$ for all $c \in (0,1)$. Let $S = \left\{c \in (0,1) \mid f(c) \neq f^*(c)\right\}$ be the set of values where $f$ and $f^*$ disagree. Assume for sake of contradiction that $S$ has positive measure. 
\begin{enumerate}
    \item The first possibility is that only a measure zero set of $c \in S$ are realizable when all users play $f$. However, this is not possible: because of background information and because the prior is initialized in $(0,1)$, the posterior eventually converges to $0$ to $1$ (but never reaches at any finite time). This means that every posterior in $(0,1)$ is realizable. 
    \item The other possibility is that a positive measure of values in $S$ are realizable all users play $f$. Let the solution $f^{**}$ be defined by $f^{**}(c) = f(c)$ for all $c$ that are realizable in $f$ and $f^{**} = f^*(c)$ otherwise. It is easy to see that $f^{**}$ would be an equilibrium in the strategic experimentation game, which is a contradiction because of uniqueness of equilibria in this game. 
\end{enumerate}
\noindent This implies that $f$ and $f^*$ agree on all but a measure $0$ set of $[0, 1]$. 

This proves that there is a unique symmetric equilibrium in the game $G$. 

We next prove that the utility achieved by this symmetric equilibrium is strictly in between the single-user optimal and the global-optimal. This follows from the fact that the utility of $A_{f^*}$ in $G$ is equal to the utility of $f^*$ in the strategic experimentation game, coupled with  Theorem \ref{thm:restatementundiscounted}. 

\end{proof}

\subsection{Proof of Theorem \ref{thm:BH}}\label{appendix:proofBH}
\begin{proof}[Proof of Theorem \ref{thm:BH}]
The proof of Theorem \ref{thm:BH} follows similarly to the proof of Theorem \ref{thm:BHundiscounted}. Like in the proof of Theorem \ref{thm:BHundiscounted}, it thus suffices to show that the equilibrium utility is a unique value $\alpha^* \in [\max_{A' \in \mathcal{A}} R_{A'}(1), \max_{A' \in \mathcal{A}} R_{A'}(N))$ at every symmetric pure strategy equilibrium in $G$. 

To analyze the game $G$, we leverage the results in the literature on strategic experimentation (see Theorem \ref{thm:restatement}). Again, the equilibrium concepts are slightly different because $f$ is a symmetric equilibrium in the strategic experimentation game if for \textit{any} choice of prior, $f$ is optimal for each user in the game with that choice of prior; on the other hand, $f$ is a symmetric equilibrium in $G$ if for the specific choice of prior of the bandit setup, $f$ is optimal for each user in the game with that choice of prior. Let $f^*$ be the unique symmetric pure strategy equilibrium in the strategic experimentation game (see property (1) in Theorem \ref{thm:restatement}). 

We first claim that the algorithm $A_{f^*}$ is a symmetric pure strategy equilibrium in $G$. To see this, notice that if a user in $G$ were to deviate to $A_{f'} \in \AAll$ for some $f'$ and achieve higher utility, then it would also be better for a user in the strategic experimentation game to deviate to $f'$ when the prior is $\mathbb{P}_{X \sim \DPrior_1}[X = 1]$. This is not possible since $f$ is an equilibrium in $G$, so $A_{f^*}$ is a symmetric pure strategy equilibrium in $G$. This establishes that an equilibrium exists in $G$. 

We next claim that if $A_f$ is a symmetric pure strategy equilibrium in $G$, then $f(c) = f^*(c)$ for all $c > c_{low}$, where $c_{low}$ is defined according to Theorem \ref{thm:restatement}. Note that by the assumption in Setup \ref{setup:continuoustimeundiscounted}, the prior is initialized above $s$, which means that it is initialized above $c_{low}$. Let $S = \left\{c > c_{low} \mid f(c) \neq f^*(c)\right\}$ be the set of values where $f$ and $f^*$ disagree. Assume for sake of contradiction that $S$ has positive measure. 
\begin{enumerate}
    \item The first possibility is that only a measure $0$ set of $c \in S$ are realizable when all users play $f$ and the prior is initialized to $\mathbb{P}_{X \sim \DPrior_1}[X = 1]$. However, this is not possible because then the trajectories of $f$ and $f^*$ would be identical so the same values would be realized for $f$ and $f^*$, but we already know that all of $c > c_{low}$ is realizable for $f^*$.
    \item The other possibility is that a positive measure of values in $S$ are realizable when all users play $f$ and the prior is initialized to $\mathbb{P}_{X \sim \DPrior_1}[X = 1]$. Let the solution $f^{**}$ be defined by $f^{**}(c) = f(c)$ for all $c$ that are realizable in $f$ and $f^{**} = f^*(c)$ otherwise. It is easy to see that $f^*$ would be an equilibrium in the game in \cite{BH98}, which is a contradiction because of uniqueness of equilibria in this game. 
\end{enumerate}
\noindent This implies that $f$ and $f^*$ agree on all but a measure $0$ set of $[c_{low}, 1]$. 

Now, let us relate the utility achieved by a symmetric equilibrium in $G$ to the utility achieved in the strategic experimentation game. 
Since the prior is initialized above $c_{high}$, property (3) in Theorem \ref{thm:restatement} tells us that the posterior never reaches $c_{low}$ but can come arbitrarily close to $c_{low}$. This in particular means if $f(c) = f^*(c)$ for all $c > c_{low}$, then the utility achieved when all users choose a strategy $f$ where is equivalent to the utility achieved when all users choose $f^*$.

Thus, it follows from property (4) in Theorem \ref{thm:restatement} that when the prior is initialized sufficiently high, the equilibrium utility is strictly smaller than the utility that users achieve in the team problem (which is equal to the global optimal utility $\max_{A' \in \mathcal{A}} R_{A'}(N)$). Moreover, it also follows from property (4) in Theorem \ref{thm:restatement} that the equilibrium utility is always at least as large as the utility $\max_{A' \in \mathcal{A}} R_{A'}(1)$ that users achieve in the single user game. 

\end{proof}

\subsection{Proof of Theorem \ref{thm:utilitysharedinfo}}\label{appendix:utilityshared}

To prove Theorem \ref{thm:utilitysharedinfo}, the key ingredient is the following fact about the game $G$.
\begin{lemma}
\label{lemma:utilitygame}
Suppose that every algorithm $A \in \mathcal{A}$ is side information monotonic. If $A$ is a symmetric pure strategy equilibrium of $G$, then the equilibrium utility at $G$ is at least $\max_{A'} R_{A'}(1)$. 
\end{lemma}
\begin{proof}[Proof of Lemma \ref{lemma:utilitygame}]
Let $A$ be a symmetric pure strategy equilibrium in $G$. To see this, notice that an user can always instead play $A^* = \argmax_{A'} R_{A'}(1)$. Since $A$ is a best response for this user, it suffices to show that playing $A^*$ results in utility at least $\max_{A'} R_{A'}(1)$. By definition, the utility that the user would receive from playing $A^*$ is $\UShared(1; \mathbf{2}_{N-1}, A^*, A)$. By side information monotonicity, we know that the presence of the background posterior updates by other users cannot reduce this user's utility, and in particular 
\[\UShared(1; \mathbf{2}_{N-1}, A^*, A) \ge R_{A^*}(1) = \max_{A'} R_{A'}(1)\] as desired. 
\end{proof}

Now we can prove Theorem \ref{thm:utilitysharedinfo} from Lemma \ref{lemma:utilitysharedinfo} and Lemma \ref{lemma:utilitygame}. 
\begin{proof}[Proof of Theorem \ref{thm:utilitysharedinfo}]
By Lemma \ref{lemma:utilitysharedinfo}, if the solution $(A,A)$ is an equilibrium for the platforms, then it is a symmetric pure strategy equilibrium of the game $G$. To lower bound $Q(A, A)$, notice that the quality level $Q(A,A)$ is equal to the utility of $A$ in the game $G$. By Lemma \ref{lemma:utilitygame}, this utility is at least $\max_{A'} R_{A'}(1)$, so $Q(A, A) \ge \max_{A'} R_{A'}(1)$ as desired. To upper bound $Q(A, A)$, notice that 
\[Q(A, A) = R_{A}(N) \le  \max_{A' \in \mathcal{A}}
 R_{A'}(N),\]
as desired.
 \end{proof}

\section{Proofs for Section \ref{sec:assumptions}}\label{appendix:assumptions}
We prove Lemmas \ref{lemma:undiscountedIM}-\ref{lemma:UR}.
\subsection{Proof of Lemma \ref{lemma:undiscountedIM}}

 \begin{proof}[Proof of Lemma \ref{lemma:undiscountedIM}]
 First, we show strict information monotonicity. Applying Lemma \ref{lemma:closedform}, we see that 
 \[R_A(n) = K(p; A, \ldots, A) = \int G(p,q) \frac{(1-A(q)) s + (q h + (1-q) l) A(q) }{\frac{\sigma^2}{\sigma_b^2} + n A(q)} dq. \]
 Note that the value $(1-A(q)) s + (q h + (1-q) l) A(q)$ is always \textit{negative} based on how we set up the rewards. We see that as $N$ increases, the denominator $N A(q)$ weakly increases at every value of $q$. This means that $R_A(n)$ is weakly increasing in $n$. To see that $R_A(n)$ \textit{strictly} increases in $n$, we note that there exists an open neighborhood around $q = 1$ where $A(q) > 0$. In this open neighborhood, we see that $n A(q)$ strictly increases in $n$, which means that $R_A(n)$ strictly increases.

 Next, we show side information monotonicity. Applying Lemma \ref{lemma:closedform}, we see that 
 \[\UShared(1; \mathbf{2}_{n}, A, A') = K(p; A, A', \ldots, A') = \int G(p,q) \frac{(1-A(q)) s + (q h + (1-q) l) A(q) }{\frac{\sigma^2}{\sigma_b^2} + A(q) + n A'(q)} dq. \] This expression is weakly larger for $n > 0$ than $n = 0$.
 
 \end{proof}
\subsection{Proof of Lemma \ref{lemma:TS}}

To prove Lemma \ref{lemma:TS}, the key technical ingredient is that if the posterior becomes more informative, then the reward increases.
\begin{lemma}[Impact of Increased Informativeness]
\label{lemma:increasedinfo}
Consider the bandit setup in Setup \ref{setup:discreteriskysafe}. Let $0 < p_{\text{prior}} < 1$ be the prior probability, and let $p'$ be the random variable for the posterior if another sample from the risky arm is observed. Let $K(p)$ denote the expected cumulative discounted reward that a user receives when they are the only user participating on a platform offering the $\epsilon$-Thompson sampling algorithm $A_{f_{\epsilon}^{TS}}$. Then the following holds: 
\[
\mathbb{E}_{p'}[K(p')] > K(p_{\text{prior}}).\]
\end{lemma}
\begin{proof}
Notice that $f_{\epsilon}^{TS}(p) = \epsilon + (1-\epsilon)p$. For notational simplicity, let $f = f_{\epsilon}^{TS}$.

Let $K_t$ denote the cumulative discounted reward for time steps $t$ to $T$ received by the user. We proceed by backwards induction on $t$ that for any $0 < p < 1$ it holds that: 
\[\mathbb{E}[K_t(p')] > K_t(p)\] where $p$ is the posterior at the beginning of time step $t$ and where $p'$ be the random variable for the posterior if another sample from the risky arm is observed before the start of time step $t$. 

The base case is $t = T$, where we see that the reward is 
\[s (1- f(q)) + (q h + (1-q) l) (f(q)) = s(1-f(q)) + (q(h-l) + l)f(q)\]
if the posterior is $q$. We see that this is a strictly convex function in $q$ for our choice of algorithm $A_f$, which means that $\mathbb{E}[K_T(p'; n)] > K_T(p; n)$ as desired. 

We now assume that the statement holds for $t+1$ and we prove it for $t$. For the purposes of our analysis, we generate a series of $2$ samples $s_1, s_2$ from the risky arm. We generate these samples recursively. Let $p_0 = p$, and let $p_1$ denote the posterior given by $p$ conditioned on $s_1$, and let $p_2$ denote the posterior conditioned on $s_1$ and $s_2$. The sample $s_{i+1}$ is drawn from a noisy observed reward for $h$ with probability $p_i$ and a noisy observed reward for $l$ with probability $p_i$. We assume that the algorithms use these samples (in order). 

Our goal is to compare two instances: Instance 1 is initialized at $p$ at time step $t$ and Instance 2 is given a sample $s_1$ before time step $t$.
The reward of an algorithm can be decomposed into two terms: the \textit{current payoff} and the \textit{continuation payoff} for remaining time steps. 

The current payoff for Instance 1 is:
\[s(1-f(p)) + (p(h-l) + l)f(p)  \]
and the current payoff for Instance 2 is:
\[\mathbb{E}[s(1-f(p_1)) + (p_1(h-l) + l)f(p_1)]. \]
Since this is a strictly convex function of the posterior, and $p_1$ is a posterior update of $p$ with the risky arm, we see that the expected current payoff for Instance 1 is strictly larger than the xpected  current payoff for Instance 2. 

The expected continuation payoff for Instance 1 is equal to the $\beta$-discounted version of: \[(1-f(p)) \cdot K_{t+1}(p) + f(p) \mathbb{E}_{s_1}[K_{t+1}(p_1)]\] and the expected continuation payoff for Instance 2 is equal to the $\beta$-discounted version of: 
\begin{align*}
   \mathbb{E}_{s_1} [f(p_1) \cdot \mathbb{E}_{s_2} [K_{t+1}(p_2)]] + (1-(f(p_1)) [K_{t+1}(p_1)]] &_{(1)} \ge  \mathbb{E}_{s_1} [f(p_1) \cdot K_{t+1}(p_1) + (1-(A(p_1)) [K_{t+1}(p_1)]] \\
   &= \mathbb{E}_{s_1} [K_{t+1}(p_1)] \\
   &\ge_{(2)} (1-f(p)) \cdot K_{t+1}(p) + f(p) \mathbb{E}_{s_1}[K_{t+1}(p_1)]
\end{align*}
where $(1)$ and $(2)$ use the induction hypothesis for $t+1$. 

This proves the desired statement.

\end{proof}

We prove Lemma \ref{lemma:TS}. 
\begin{proof}[Proof of Lemma \ref{lemma:TS}]
We first show strict information monotonicity and then we show side information monotonicity. Our key technical ingredient is Lemma \ref{lemma:increasedinfo}.

\paragraph{Proof of strict information monotonicity.} We show that for any $n \ge 1$, it holds that $R_A(n+1) > R_A(n)$. To show this, we construct a sequence of $T+1$ ``hybrids''. In particular, for $0 \le t \le T$, let the $t$th hybrid correspond to the bandit instance where $2$ users participate in the algorithm in the first $t$ time steps and $1$ users participate in the algorithm in the remaining time steps. These hybrids enable us to isolate and analyze the gain of one additional observed sample. 

For each $2 \le t \le T$, it suffices to show that the $t-1$th hybrid incurs larger cumulative discounted reward than the $t-2$th hybrid. Notice a user participating in both of these hybrids at all time steps achieves the same expected reward in the first $t-1$ time steps for these two hybrids. The first time step where the two hybrids deviate in expectation is the $t$th time step (after the additional information from the $t-1$th hybrid has propagated into the information state). Thus, it suffices to show that $t-1$th hybrid incurs larger cumulative discounted reward than the $t-2$th hybrid between time steps $t$ and $T$. Let $\mathcal{H}$ be the history of actions and observed rewards from the first $t-2$ time steps, the arm and the reward for the 1st user at the $t-1$th time step. We condition on $\mathcal{H}$ for this analysis. Let $a$ denote the arm chosen by the $2$nd user at the $t-1$th time step in the $t-1$th hybrid. We split into cases based on $a$. 

The first case is that $a$ is the safe arm. In this case, the hybrids have the same expected reward for time steps $t$ to $T$. 

The second case is $a$ is the risky arm. This happens with nonzero probability given $\mathcal{H}$ based on the structure of $\epsilon$-Thompson sampling and based on the structure of the risky-safe arm problem. We condition on $\mathcal{H}' = \mathcal{H} \cup \left\{a = \text{risky}\right\}$ for the remainder of the analysis and show that the $t-1$th hybrid achieves strictly higher reward than the $t$th hybrid. 

To show this, let $K_t(q)$ denote the cumulative discounted reward incurred from time step $t$ to time step $T$ if the posterior at the start of time step $t$ is $q$. Let $p$ be the posterior given by conditioning on $\mathcal{H}$ (which is the same as conditioning on $\mathcal{H}'$). Let $p'$ denote the distribution over posteriors given by updating with an additional sample from the risky arm. Showing that he $t-1$th hybrid achieves strictly higher reward than the $t$th hybrid reduces to showing that $\mathbb{E}[K_t(p')] > K_t(p)$. Since the discounting is geometric, this is equivalent to showing that $\mathbb{E}[K_1(p')] > K_1(p)$, which follows from Lemma \ref{lemma:increasedinfo} as desired.

\paragraph{Proof of side information monotonicity.} We apply the same high-level approach as in the proof of  strict information monotonicity, but construct a different set of hybrids. For $0 \le t \le T$, let the $t$th hybrid correspond to the bandit instance where both the user who plays algorithm $A$ and the other user also updates the shared information state with the algorithm $A'$ in the first $t$ time steps and only the single user playing $A$ updates the information state in the remaining time steps. 

For each $2 \le t \le T$, it suffices to show that the $t-1$th hybrid incurs larger cumulative discounted reward than the $t-2$th hybrid. Notice a user participating in both of these hybrids at all time steps achieves the same expected reward in the first $t-1$ time steps for these two hybrids. The first time step where the two hybrids deviate in expectation is the $t$th time step (after the additional information from the $t-1$th hybrid has propagated into the information state). Thus, it suffices to show that $t-1$th hybrid incurs larger cumulative discounted reward than the $t-2$th hybrid between time steps $t$ and $T$. Let $\mathcal{H}$ be the history of actions and observed rewards from the first $t-2$ time steps, the arm and the reward for the 1st user at the $t-1$th time step. We condition on $\mathcal{H}$ for this analysis. Let $a$ denote the arm chosen by the $2$nd user at the $t-1$th time step in the $t-1$th hybrid. We split into cases based on $a$. 

The first case is that $a$ is the safe arm. In this case, the hybrids have the same expected reward for time steps $t$ to $T$. 

The second case is $a$ is the risky arm. If this happens with zero probability conditioned on $\mathcal{H}$, we are done. Otherwise, we condition on $\mathcal{H}' = \mathcal{H} \cup \left\{a = \text{risky}\right\}$ for the remainder of the analysis and show that the $t-1$th hybrid achieves higher reward than the $t$th hybrid. 

To show this, let $K_t(q)$ denote the cumulative discounted reward incurred from time step $t$ to time step $T$ if the posterior at the start of time step $t$ is $q$. Let $p$ be the posterior given by conditioning on $\mathcal{H}$ (which is the same as conditioning on $\mathcal{H}'$). Let $p'$ denote the distribution over posteriors given by updating with an additional sample from the risky arm. Showing that he $t-1$th hybrid achieves strictly higher reward than the $t$th hybrid reduces to showing that $\mathbb{E}[K_t(p')] > K_t(p)$. Since the discounting is geometric, this is equivalent to showing that $\mathbb{E}[K_1(p')] > K_1(p)$, which follows from Lemma \ref{lemma:increasedinfo} as desired.

\end{proof}

\subsection{Proof of Lemma \ref{lemma:undiscountedUR}}

We prove Lemma \ref{lemma:undiscountedUR}.
\begin{proof}[Proof of Lemma \ref{lemma:undiscountedUR}]
Applying Lemma \ref{lemma:closedform}, we see that:
 \[R_A(n) = K(p; A, \ldots, A) = \int G(p,q) \frac{(1-A(q)) s + (q h + (1-q) l) A(q) }{\frac{\sigma^2}{\sigma_b^2} + n A(q)} dq. \]
 
It follows immediately that the set $\left\{R_A(N) \mid A \in \AAllContinuous\right\}$ is connected. To see that the supremum is achieved, we use the characterization in \citet{BHSimple} of the team optimal as a cutoff strategy within $\AAllContinuous$. It thus suffices to show that there exists $A'$ such that $R_{A'}(N) \le \max_{A \in \mathcal{A}} R_A(1)$. To see this, let $f_{1-\epsilon}$ be the cutoff strategy where $f_{1-\epsilon}(p) = 1$ for $p \ge 1 - \epsilon$  and $f_{1-\epsilon}(p) = 0$ for $p < 1-\epsilon$. We see that there exists $\epsilon > 0$ such that $R_{A_{f_{1-\epsilon}}}(N) \le \max_{A \in \mathcal{A}} R_A(1)$. 

\end{proof}

\subsection{Proof of Lemma \ref{lemma:UR}}

We prove Lemma \ref{lemma:UR}.
\begin{proof}[Proof of Lemma \ref{lemma:UR}]

First, we show that there exists $A' \in \mathcal{A}$ such that $R_{A'}(N) \le \max_{A \in \mathcal{A}} R_A(1)$. To see this, notice that:
\[R_{A_1}(N) = R_{A_1}(1) \le \max_{A \in \mathcal{A}_{closure}} R_A(1) \] since the reward of uniform exploration is independent of the number of other users. 

It now suffices to show that the set $\left\{R_{A_\epsilon}(N) \mid \epsilon \in [0,1] \right\}$ is closed for every $A \in \mathcal{A}$. 

To analyze the expected $\beta$-discounted utility of an algorithm, it is convenient to formulate it in terms of the sequences of actions and rewards observed by the algorithm. Let the realized history denote the sequence $\mathcal{H} = (a^1_1, o^1_1), \ldots, (a^N_1, o^N_1), \ldots, (a^1_T, o^1_T), \ldots, (a^N_T, o^N_T)$ of (pulled arm, observed reward) pairs observed at each time step. An algorithm $A'$ induces a distribution $\mathcal{D}_{A'}$ over realized histories (that depends on the prior distributions $\DPrior_i$). If we let 
\[f((a^1_1, o^1_1), \ldots, (a^N_1, o^N_1), \ldots, (a^1_T, o^1_T), \ldots, (a^N_T, o^N_T)) := \sum_{t=1}^T r_{a^1_t} \beta^t,\] then the expected $\beta$-discounted cumulative reward of an algorithm $A'$ is: 
\begin{align*}
 R_A(N) &= \mathbb{E}_{(a^1_1, o^1_1), \ldots, (a^N_1, o^N_1), \ldots, (a^1_T, o^1_T), \ldots, (a^N_T, o^N_T) \sim \mathcal{D}_{A'}} \left[\sum_{t=1}^T r_{a^1_t} \beta^t \right] \\
 &= \mathbb{E}_{\mathcal{H} \sim \mathcal{D}_{A'}} \left[f(\mathcal{H}) \right].
\end{align*}
Since the mean rewards are bounded, we see that: 
\[f((a^1_1, o^1_1), \ldots, (a^N_1, o^N_1), \ldots, (a^1_T, o^1_T), \ldots, (a^N_T, o^N_T))\in \left[\min_{1 \le i \le k} r_i, \max_{1 \le i \le k} r_i\right].\]

Now, notice that the total-variation distance 
\[TV(\mathcal{D}_{A_{\epsilon_1}}, \mathcal{D}_{A_{\epsilon_2}}) \le 1 - (\epsilon_1-\epsilon_2)^{NT}\] because with $(\epsilon_1-\epsilon_2)^{NT}$ probability, $A_{\epsilon_1}$ behaves identically to $A_{\epsilon_2}$. We can thus conclude that:
\[|R_{A_{\epsilon_1}}(N) - R_{A_{\epsilon_2}}(N)| \le \left|\max_{1 \le i \le k} r_i - \min_{1 \le i \le k} r_i\right| TV(\mathcal{D}_{A_{\epsilon_1}}, \mathcal{D}_{A_{\epsilon_2}}) \le \left|\max_{1 \le i \le k} r_i - \min_{1 \le i \le k} r_i\right|  \left(1 - (\epsilon_1-\epsilon_2)^{NT}\right).\]
This proves that $R_{A_\epsilon}$ changes continuously in $\epsilon$ which proves the desired statement. 
\end{proof}

\end{document}